\documentclass[conference]{IEEEtran}
\usepackage{amsfonts}
\usepackage{amsmath}
\usepackage{cases}
\usepackage{mathrsfs}
\usepackage{bbm}
\usepackage{amsfonts}
\usepackage{multirow}
\usepackage{caption}
\usepackage{color}
\usepackage{flushend}
\usepackage{float}
\usepackage[ruled,vlined,linesnumbered]{algorithm2e}
\usepackage{algorithmic}
\usepackage{subfigure}
\usepackage{extarrows}
\usepackage{graphicx,graphics,txfonts}
\usepackage{array}%tabularx
\usepackage{tabularx}
\usepackage{stmaryrd}
\begin{document}

\newtheorem{definition}{Definition}
\renewcommand{\algorithmicrequire}{\textbf{Requires}}
\newtheorem{theorem}{Theorem}
\newtheorem{lemma}{Lemma}
\newtheorem{axiom}{Axiom}
\newtheorem{example}{Example}
\newtheorem{corollary}{Corollary}
\newtheorem{property}{Property}

\newcommand{\partitle}[1]{\medskip \noindent \textbf{#1.}}
\newcommand{\subpartitle}[1]{\medskip \emph{#1.}}
\newcommand{\topcaption}{%
\setlength{\abovecaptionskip}{0pt}%
\setlength{\belowcaptionskip}{0pt}%
\caption}

\title{Eclipse: Generalizing kNN and Skyline}

\author{

\IEEEauthorblockN{Jinfei Liu\IEEEauthorrefmark{1}\IEEEauthorrefmark{2}, Li Xiong\IEEEauthorrefmark{1}, Qiuchen Zhang\IEEEauthorrefmark{1}, Jian Pei\IEEEauthorrefmark{3}, Jun Luo\IEEEauthorrefmark{4}}
    \IEEEauthorblockA{\IEEEauthorrefmark{1}Department of Mathematics \& Computer Science, Emory University
    \\\{jinfei.liu, lxiong, qiuchen.zhang\}@emory.edu}
    \IEEEauthorblockA{\IEEEauthorrefmark{2}College of Computing, Georgia Institute of Technology, jinfei.liu@cc.gatech.edu}
    \IEEEauthorblockA{\IEEEauthorrefmark{3}JD.com \& Simon Fraser University, jpei@cs.sfu.ca}
    \IEEEauthorblockA{\IEEEauthorrefmark{4}Machine Intelligence Center, Lenovo \& SIAT, Chinese Academy of Sciences, jun.luo@siat.ac.cn}
}

\maketitle
%---------------------------------------------------------------------------------------------------------------------------------------------------------------------------------------------------------------------------------------------------------------
%---------------------------------------------------------------------------------------------------------------------------------------------------------------------------------------------------------------------------------------------------------------

%-----------------------------------------------------------------------------------------------------------------------------------------------
\begin{abstract}
$k$ nearest neighbor ($k$NN) queries and skyline queries are important operators on multi-dimensional data points. Given a query point, $k$NN query returns the $k$ nearest neighbors based on a scoring function such as a weighted sum of the attributes, which requires predefined attribute weights (or preferences). Skyline query returns all possible nearest neighbors for any monotonic scoring functions without requiring attribute weights but the number of returned points can be prohibitively large. We observe that both $k$NN and skyline are inflexible and cannot be easily customized.

In this paper, we propose a novel \emph{eclipse} operator that generalizes the classic $1$NN and skyline queries and provides a more flexible and customizable query solution for users. In eclipse, users can specify rough and customizable attribute preferences and control the number of returned points. We show that both $1$NN and skyline are instantiations of eclipse. To process eclipse queries, we propose a baseline algorithm with time complexity $O(n^22^{d-1})$, and an improved $O(n\log ^{d-1}n)$ time transformation-based algorithm, where $n$ is the number of points and $d$ is the number of dimensions. Furthermore, we propose a novel index-based algorithm utilizing duality transform with much better efficiency. The experimental results on the real NBA dataset and the synthetic datasets demonstrate the effectiveness of the eclipse operator and the efficiency of our eclipse algorithms.
\end{abstract}

%-------------------------------------------------------------------------------------------------------------------------------------------------
%--------------------------------------------------------------------------------------------------------------------------------------------------
\section{Introduction}\label{sec:Introduction}

Both $k$NN queries and skyline queries are important operators with many applications in computer science.  Given a dataset $P$ of multi-dimensional objects or points, $k$NN returns the $k$ points closest to a query point given a scoring function. One commonly used scoring function is \emph{weighted sum} of the attributes, where the weights indicate the importance of attributes. For a point $p$ in a $d$ dimensional space $p=(p[1],p[2],...,p[d])$, given an attribute weight vector $\textbf{w}=\langle w[1],w[2],...,w[d]\rangle$, the weighted sum $S(p)$ is defined as $S(p)=\sum_{i=1}^dp[i]w[i]$ (assuming the origin is the query point). $k$NN returns the $k$ points with smallest $S(p)$. When $k=1$, we also call it $1$NN query.

One drawback of $1$NN (and $k$NN in general) is its dependence on the exact scoring function. Skyline query is an alternative solution involving multi-criteria decision making without relying on a specific scoring function. Considering the origin as the query point again, skyline consists all Pareto-nearest points that are not dominated by any other point, and a point $p$ dominates another point $p' (p\neq p')$  if it is at least closer to the query point on one dimension and as close on all the other dimensions.

It has been recognized that $1$NN and skyline each has its advantage but also comes with a cost: 1) $1$NN returns the exact nearest neighbor but depends on a predefined attribute weight vector which can be too specific in practice; 2) skyline returns all possible nearest neighbors without requiring any attribute weight vector but the number of returned points can be prohibitively large, in the worst case, the whole data set being returned. It is desirable to have a more flexible and customizable generalization to satisfy users' diverse needs.

\partitle{Motivating example}
Assume a conference organizer needs to recommend a set of hotels for conference participants based on the distance to the conference site and price. Note that we use two dimensional case in our running examples. Figure \ref{fig:1nn} shows a dataset $P=\{p_1,p_2,...,p_4\}$, each representing a hotel with the two attributes. The organizer can use $k$NN queries and specify an attribute weight vector such as $\textbf{w}=\langle 2,1\rangle$ or attribute weight ratio $r=w[1]/w[2]=2$ indicating distance is twice more important than price. In this case, $p_1$ has the smallest score $S(p_1)=8$ and is the nearest neighbor. We can also visualize $S(p)$ in a two dimensional space. If we draw a score line over a point $p$ with slope $-2$, $S(p)$ is essentially its y-intercept. The nearest neighbor is hence $p_1$ which has the smallest y-intercept as shown in Figure \ref{fig:1nn}(b).

\vspace{-1em}
\begin{figure}[htb]
 \centering
 \includegraphics[width=0.35\textwidth]{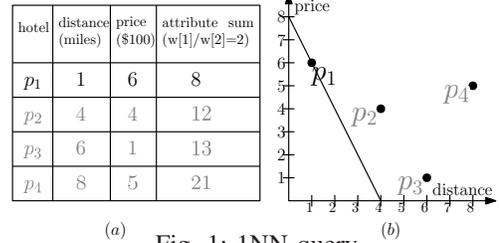}
  \vspace{-1em}
 \caption{$1$NN query.}
 \label{fig:1nn}
\end{figure}
\vspace{-1em}

The organizer can also use skyline queries to retrieve all hotels not dominated by others since the preferences of the participants are unknown. Figure \ref{fig:Skyline} shows the same dataset as in Figure \ref{fig:1nn}.  Given a point $p$, any point lying on the upper right corner of $p$ is dominated by $p$. Hence, $p_1$, $p_2$, and $p_3$ are the skyline points as they are not dominated by any other point.

\vspace{-1em}
\begin{figure}[htb]
 \centering
 \includegraphics[width=0.35\textwidth]{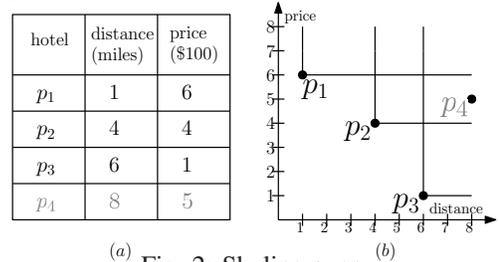}
  \vspace{-1em}
 \caption{Skyline query.}
 \label{fig:Skyline}
\end{figure}
\vspace{-1em}

We can see that both solutions above are inflexible and cannot be customized. The definitions are either too specific to one ranking function (as in $1$NN) or taking no ranking function at all (as in skyline). \emph{If the organizer knows that price is more important than distance to all student participants but the relative importance varies from one student to another, neither $1$NN nor skyline can incorporate this kind of preference ``range''}. One may think that $k$NN provides certain flexibility or relaxation for $1$NN by selecting the top $k$ solutions. However, it only provides flexibility on the ``depth'' in that it selects the next best solutions with respect to the same exact attribute weights. What is desired here is the flexibility on the ``breadth'' of the attribute weights which can capture a user's rough preference among the attributes.

\partitle{Contributions}
In this paper, we propose a novel notion of eclipse\footnote{``eclipse'' comes from solar eclipse and lunar eclipse, we take its notion of ``partial'' to highlight the customizability of our query definition.} that generalizes the classic 1NN and skyline queries while providing a more flexible and customizable query solution for users. We first show a generalized definition of $1$NN and skyline through the notion of domination and then introduce our eclipse definition.

For the $1$NN query, we can say $p$ $1$NN-dominates $p'$, if $S(p)< S(p')$ for a given attribute weight ratio $r=l$ (or $r\in [l,l]$). We assume two dimensional space for example here. For the skyline query, we can easily see that if $p$ dominates (we explicitly say skyline-dominates to differentiate from $1$NN-dominates) $p'$, we have $S(p)\leq S(p')$ for all attribute weight ratio $r\in [0,+\infty)$ given a linear scoring function or any monotonic scoring function. In other words, any point lying on the upper right half of the score line of $p$ (with a flat angle) is $1$NN-dominated by $p$ (Figure \ref{fig:1nn}). Any point lying on the upper right quadrant of $p$ (with a right angle) is skyline-dominated by $p$ (Figure \ref{fig:Skyline}). Both $1$NN and skyline can be defined as those points that are not dominated by any other point. Table \ref{tab:definitions} shows the comparison.

\vspace{-1em}
\begin{table}[htb]\centering
\caption{Definitions.}\label{tab:definitions}
\vspace{-1em}
{%
\begin{tabular}{|c|c|c|}
\hline
                         & domination weight ratio & domination range\\
\hline
$1$NN & $[l,l]$ & flat angle\\
\hline
skyline    & $[0,+\infty)$ & right angle \\
\hline
eclipse & $[l,h]$ & obtuse angle\\
\hline
\end{tabular}}
\end{table}%
\vspace{-1em}

Based on this generalized notion of dominance, we propose the eclipse query. We say $p$ eclipse-dominates $p'$, if $S(p)\leq S(p')$ for all $r \in [l,h]$, where $[l,h]$ is a range for the attribute weight ratio. The eclipse points are those points that are not eclipse-dominated by any other point in $P$. Intuitively, the range of $[l,h]$ for the attribute weight ratio allows users to have a flexible and approximate preference of the relative importance of the attributes rather than an exact value $l$ (as in $1$NN) or an infinite interval $[0,+\infty)$ (as in skyline). As a result, eclipse combines the best of both $1$NN and skyline and returns a subset of points from skyline which are possible nearest neighbors for all scoring functions with attribute weight ratio in the given range. We can easily see that both $1$NN and skyline are instantiations of eclipse queries.

Recall our running example. If the conference organizers want to incorporate the preference that price is more important than distance for all student participants, they can set the attribute weight ratio as $r\in [0,1)$. For practical usage, in order to reduce the burden of parameter selection for users, we envision that users can either specify an attribute weight vector as in $k$NN which can be relaxed into ranges with a margin, or specify the relative importance of the attributes in categorical values such as very important, important, similar, unimportant, very unimportant, which correspond to predefined attribute weight ranges.

\vspace{-1em}
\begin{figure}[htb]
 \centering
 \includegraphics[width=0.35\textwidth]{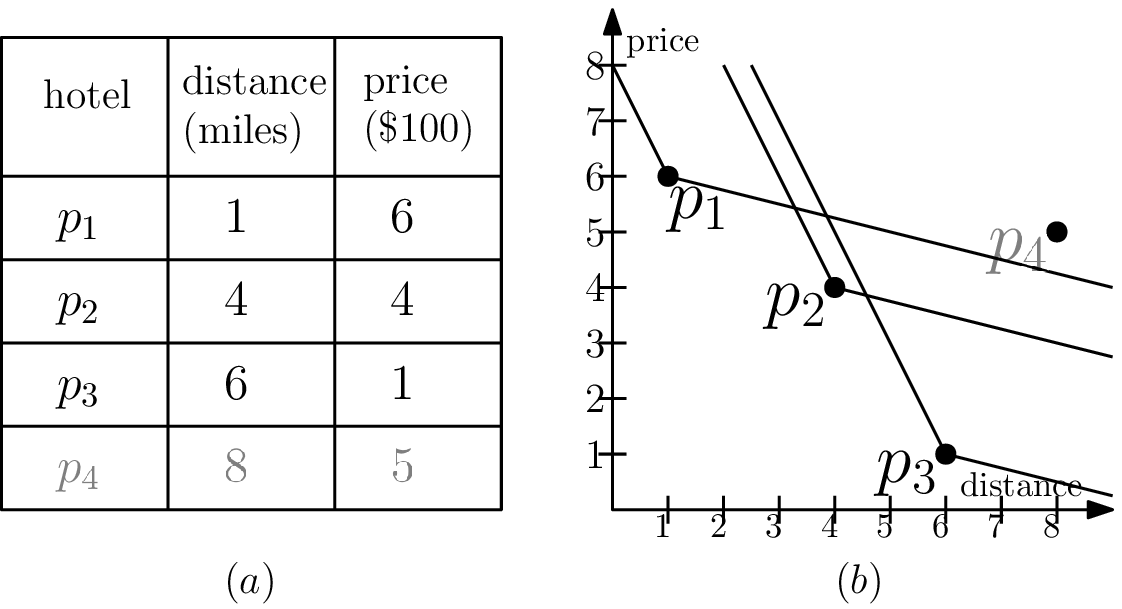}
  \vspace{-1em}
 \caption{Eclipse query.}
 \label{fig:eclipseQ}
\end{figure}
\vspace{-1em}

Figure \ref{fig:eclipseQ} shows another example of eclipse query with an attribute weight ratio range $r \in [1/4,2]$ which indicates distance is relatively comparable to price. For a point $p$, we can draw two score lines with slopes $-h$ and $-l$, respectively. A point eclipse-dominates other points in its upper right range (with an obtuse angle) between the two score lines. We can see that $p_4$ is eclipse-dominated by $p_1,p_2,p_3$. The eclipse query returns $p_1$, $p_2$, and $p_3$ as they cannot eclipse-dominate each other. Please note that $p_1$ cannot skyline-dominate $p_4$ in skyline, but $p_1$ can eclipse-dominate $p_4$ in eclipse. Therefore, eclipse generally returns fewer points than skyline as the domination range is larger than skyline.

It is non-trivial to compute eclipse points efficiently. To determine if point $p$ eclipse-dominates $p'$, we need to check if $S(p)\leq S(p')$ for all $r \in [l,h]$ but there are an infinite number of testings if we do it one by one. We first prove that we only need to check the boundary values $l$ and $h$ rather than the entire range. Given the boundary values, a straightforward algorithm to determine the eclipse-dominance relationship for each pair of points leads to $O(n^22^{d-1})$ time complexity. We propose an algorithm by transforming the eclipse problem to the skyline problem, which leads to much better $O(n\log ^{d-1}n)$ time complexity. In addition, we propose a novel index-based algorithm utilizing duality transform with further improved efficiency. The main idea is to build two index structures, Order Vector Index and Intersection Index, to allow us to quickly compute the dominance relationships between the points based on a given attribute weight ratio range. To implement Intersection Index in high dimensional space, we propose line quadtree and cutting tree with different tradeoffs in terms of average case and worst case performance.

We briefly summarize our contributions as follows.
\vspace{-0.7em}
\begin{itemize}
\item We propose a novel operator \emph{eclipse} that generalizes the classic $1$NN and skyline queries while providing a more flexible and customizable query solution for users.  We formally show its properties and its relationship with other related notions including $1$NN, convex hull, and skyline.  We show that $1$NN and skyline are the special cases of eclipse.\vspace{-0.7em}

\item We present an efficient $O(n\log ^{d-1}n)$ time transformation-based algorithm for computing eclipse points by transforming the eclipse problem to the skyline problem.\vspace{-0.7em}

\item We present an efficient index-based algorithm by utilizing index structures and duality transform with detailed complexity analysis shown in Section \ref{sec:index-based}.\vspace{-0.7em}

\item We conduct comprehensive experiments on the real and synthetic datasets. The experimental results show that eclipse is interesting and useful, and our proposed algorithms are efficient and scalable.
\end{itemize}
\vspace{-1em}

\partitle{Organization} The rest of the paper is organized as follows. Section \ref{sec:definition} introduces the eclipse definition, eclipse properties, and the relationship between the eclipse query and other queries. We present the transformation-based algorithms for computing eclipse points in Section \ref{sec:transformation-based}, and the index-based algorithms in Section \ref{sec:index-based}. We report the experimental results and findings in Section \ref{sec:Experiments}. Section \ref{sec:Related} presents the related work. Section \ref{sec:Conclusion} concludes the paper.

%--------------------------------------------------------------------------------------------------------------------------------------------------

%--------------------------------------------------------------------------------------------------------------------------------------------------
\section{Definitions and Properties}\label{sec:definition}

In this section, we first show some preliminaries and then give the formal definition of eclipse as well as a few properties of eclipse. Finally, we show the relationship between the eclipse query and others queries. For reference, a summary of notations is given in Table \ref{tab:notations}.

\vspace{-2em}
\begin{table}[htb]\centering
\caption{The summary of notations.}\tiny\label{tab:notations}
\vspace{-1em}
{%
\footnotesize
\begin{tabular}{|c|c|}
\hline
Notation & Definition\\
\hline
$p_i$ & $i^{th}$ point in dataset $P$\\
\hline
$p_i[j]$ & $j^{th}$ dimension of point $p_i$\\
\hline
$p\prec_s p'$ & p skyline-dominates p'\\
\hline
$p\prec_e p'$ & p eclipse-dominates p'\\
\hline
$p\prec_1 p'$ & p $1$NN-dominates p'\\
\hline
$n$ & number of points\\
\hline
$u$ & number of skyline points\\
\hline
$d$ & number of dimensions\\
\hline
$w[j]$ & $j^{th}$ attribute weight\\
\hline
$\textbf{w}=\langle w[1],...,w[d]\rangle$ & attribute weight vector\\
\hline
$r[j]=w[j]/w[d]$ & $j^{th}$ attribute weight ratio\\
\hline
$\textbf{r}=\langle r[1],...,r[d-1]\rangle$ & attribute weight ratio vector\\
\hline
$S(p)$ & weighted sum of point $p$\\
\hline
$S(p)_{\textbf{r}}$ & weighted sum of point $p$ for $\textbf{r}$ \\
\hline
$c_i$ & mapped corresponding point of $p_i$\\
\hline
$c_i[j]$ & $c_i$ on the $j^{th}$ dimension\\
\hline
\end{tabular}}
\end{table}%
\vspace{-1em}

\subsection{Preliminaries}

In the traditional $1$NN definition, $1$NN returns the point closest to a query point according to a given attribute weight vector. In the traditional skyline definition \cite{DBLP:journals/pvldb/LiuXPLZ15}, skyline returns all points that are not dominated by any other point. We can generalize the traditional $1$NN and skyline definitions using the notion of domination. For $1$NN, we say $p$ $1$NN-dominates $p'$, if $S(p)<S(p')$ for a given attribute weight vector. For skyline, we say $p$ skyline-dominates $p'$, if $S(p)\leq S(p')$ for all attribute weight ratio $r\in [0,+\infty)$ given a linear scoring function. We provide the generalized definitions for $1$NN and skyline as follows.

\begin{definition}(\textbf{$1$NN}).
Given a dataset $P$ of $n$ points in $d$ dimensional space and an attribute weight vector $\textbf{w}=\langle w[1],w[2],...,w[d]\rangle$ or an attribute weight ratio vector $\textbf{r}=\langle r[1],r[2],...,r[d-1]\rangle$, where $r[j]=w[j]/w[d]$. Let $p=(p[1],p[2],...,p[d])$ and $p'=(p'[1], p'[2],..., p'[d])$ be two different points in $P$, we say $p$ $1$NN-dominates $p'$, denoted by $p\prec_1 p'$, if $S(p)< S(p')$ for $\boldsymbol{r[j]\in [l_j,l_j]}$, where $l_j$ is a user-specified value for the $j^{th}$ attribute weight ratio,  $S(p)=\sum_{i=1}^dp[i]w[i]$ is the weighted sum\footnote{In this paper, for the purpose of simplicity, we focus on $L_1$ norm. However, we note that the algorithms for $L_1$ norm can be easily extended to $L_p$ norm ($L_p^W(p_i)=(\sum_{i=1}^dw[i]p_i^p)^{1/p}$) for $p\geq2$. The reasons are 1) factor $1/p$ in the exponent part cannot affect the ranking of each point and 2) there is no difference to compute $w[i]p_i^p$ for different $p$.} of $p$, $l_j\in [0,+\infty)$, and $j=1,2,...,d-1$. The $1$NN point is the point that is not dominated by any other point in $P$.
\end{definition}

\begin{definition}(\textbf{Skyline}).
Given a dataset $P$ of $n$ points in $d$ dimensional space and an attribute weight vector $\textbf{w}=\langle w[1],w[2],...,w[d]\rangle$ or an attribute weight ratio vector $\textbf{r}=\langle r[1],r[2],...,r[d-1]\rangle$, where $r[j]=w[j]/w[d]$. Let $p=(p[1],p[2],...,p[d])$ and $p'=(p'[1], p'[2],..., p'[d])$ be two different points in $P$, we say $p$ dominates $p'$, denoted by $p\prec_s p'$, if $S(p)\leq S(p')$ for all $\boldsymbol{r[j]\in [0,+\infty)}$, where $j=1,2,...,d-1$. The skyline points are those points that are not dominated by any other point in $P$.
\end{definition}

\subsection{Eclipse Definition and Properties}

We define a new eclipse query below which allows a user to define an attribute weight ratio range for the domination.

\begin{definition}(\textbf{Eclipse}).
Given a dataset $P$ of $n$ points in $d$ dimensional space and an attribute weight vector $\textbf{w}=\langle w[1],w[2],...,w[d]\rangle$ or an attribute weight ratio vector $\textbf{r}=\langle r[1],r[2],...,r[d-1]\rangle$, where $r[j]=w[j]/w[d]$. Let $p=(p[1],p[2],...,p[d])$ and $p'=(p'[1], p'[2],..., p'[d])$ be two different points in $P$, we say $p$ eclipse-dominates $p'$, denoted by $p\prec_e p'$, if $S(p)\leq S(p')$ for all $\boldsymbol{r[j]\in [l_{j},h_{j}]}$, where $[l_j,h_j]$ is a user-specified range for the $j^{th}$ attribute weight ratio. The eclipse points are those points that are not dominated by any other point in $P$.
\end{definition}

\begin{example}
In the $1$NN query of Figure \ref{fig:1nn}, given the attribute weight ratio vector $\textbf{r}= \langle 2\rangle$, $p_1$ dominates the points in the flat angle range, i.e., $p_2$, $p_3$, and $p_4$. In the skyline query of Figure \ref{fig:Skyline}, $p_1$ dominates the points in the right angle range, and $p_1$ cannot dominate any point in this example. In the eclipse query of Figure \ref{fig:eclipseQ}, given the attribute weight ratio vector $\textbf{r}=\langle r\rangle$, where $r\in [1/4,2]$, $p_1$ dominates the points in the obtuse angle range, i.e., $p_4$.
\end{example}

For each point, we call the range that it can dominate as \emph{domination range}, the boundary line as \emph{domination line} (\emph{domination hyperplane} in high dimensional space), and the attribute weight ratio vector that determines the domination line plus $r[d]=w[d]/w[d]=1$ as \emph{domination vector}. For example, for point $p_1$ in Figure \ref{fig:eclipseQ}, given the attribute weight ratio range $r\in [1/4,2]$, the domination lines are $y=-2x+8$ and $y=-1/4x+6.25$, the domination vectors are $\langle 2,1 \rangle$ and $\langle 1/4,1 \rangle$, and the domination range is the obtuse angle range on the upper right corner of the domination lines.

We show several properties of eclipse queries below.

\begin{property}(\textbf{Asymmetry}).\label{prop:asy}
Given two points $p$ and $p'$, if $p\prec _ep'$, then $p'\nprec _ep$.
\end{property}

\begin{proof}
Because $p\prec_ep'$, for any attribute weight ratio vector $\textbf{r}=\langle r[1],r[2],...,r[d-1]\rangle$, where $r[j]\in [l_{j},h_{j}]$ for j=1,2,...,d-1, we have $S(p)_{\textbf{r}}\leq S(p')_{\textbf{r}}$, where $S(p)_{\textbf{r}}$ is the weighted sum for $\textbf{r}$. Therefore, $p'\nprec _ep$.
\end{proof}

\begin{property}(\textbf{Transitivity}).\label{prop:tran}
Given three points $p_1$, $p_2$, and $p_3$, if $p_1\prec _ep_2$ and $p_2\prec _ep_3$, then $p_1\prec _ep_3$.
\end{property}

\begin{proof}
If we have $p_1\prec _ep_2$ and $p_2\prec _ep_3$, for any attribute weight ratio vector $\textbf{r}=\langle r[1],r[2],...,r[d-1]\rangle$, where $r[j]\in [l_{j},h_{j}]$ for j=1,2,...,d-1, we have $S(p_1)_{\textbf{r}}\leq S(p_2)_{\textbf{r}}$ and $S(p_2)_{\textbf{r}}\leq S(p_3)_{\textbf{r}}$. Therefore, we have $S(p_1)_{\textbf{r}}\leq S(p_3)_{\textbf{r}}$, that is $p_1\prec _ep_3$.
\end{proof}

Furthermore, we show the dominance definition in skyline is stricter than the dominance definition in eclipse by the following two properties.

\begin{property}\label{pro:1}
If $p\prec_s p'$, then $p\prec_e p'$, vice is not.
\end{property}
\vspace{-0.5em}
\begin{property}\label{pro:2}
If $p\nprec_s p'$, it is possible that $p\prec_e p'$.
\end{property}

%\begin{property}
%If $p\nprec_e p'$, it is impossible that $p\prec_s p'$.
%\end{property}

\subsection{Relationship with other Definitions}
In this subsection, we discuss the relationship between the eclipse query and other classic queries, i.e., $1$NN, convex hull, and skyline. We note that the convex hull query returns the points from origin's view rather than the entire traditional convex hull. For example, in Figure \ref{fig:1nn}, the convex hull query returns $p_1,p_3$ rather than $p_1,p_3,p_4$.

The relationship among $1$NN, convex hull, eclipse, and skyline is shown in Figure \ref{fig:relationship}. $1$NN returns the best one point given a linear scoring function with specific weight for each attribute. Convex hull returns the best points given any linear scoring functions, so convex hull contains all possible $1$NN points. Skyline returns the best points given any monotone scoring functions. Eclipse returns the best points given a linear scoring function with a weight range for each attribute. As a result, skyline is the superset of eclipse and convex hull, $1$NN contains a point that belongs to the result set of all other queries. Depending on the range, eclipse ($[l,h]$) can be instantiated to be $1$NN ($[l,l]$) or skyline ($[0,+\infty)$). Therefore, eclipse not only contains some points that belong to convex hull but also some points that do not belong to convex hull.

\begin{figure}[htb]
 \centering
 \includegraphics[width=0.25\textwidth]{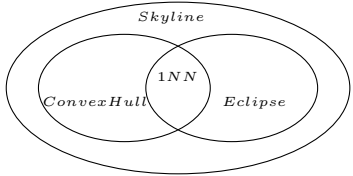}
 \caption{Relationship between eclipse and other definitions.}
 \label{fig:relationship}
\end{figure}
\vspace{-1em}

\section{Transformation-based Algorithms}\label{sec:transformation-based}

In this section, we first show a baseline algorithm in Subsection \ref{subsec:baseline} and then show an improved algorithm by transforming the eclipse problem to the skyline problem in two dimensional space in Subsection \ref{subsec:transformation-based2D} and high dimensional space in Subsection \ref{subsec:transformation-basedhighD}.

\subsection{Baseline Algorithm}\label{subsec:baseline}
In order to check the dominance between a point $p$ and other points in two dimensional space, we observe that instead of computing all the continuous values in the range $[l_j,h_j]$ for $S(p)$, we only need to compute the boundary values of the range for $S(p)$. We note that although \cite{DBLP:journals/pvldb/CiacciaM17} presented a similar algorithm, they did not give any proof for the correctness.

\begin{theorem}\label{the:endpoint2D}
Given an attribute weight vector $\textbf{r}=\langle r\rangle$, where $r\in [l,h]$, if $S(p)_{\textbf{r}}\leq S(p')_{\textbf{r}}$ for $r=l$ and $r=h$, we have $S(p)_{\textbf{r}}\leq S(p')_{\textbf{r}}$ for all $r\in [l,h]$, where $S(p)_{\textbf{r}}$ is the weighted sum of point $p$ for $\textbf{r}$.
\end{theorem}

\begin{proof}
Because for $r=l$ and $r=h$, $S(p)_{\textbf{r}}\leq S(p')_{\textbf{r}}$, we have $lp[1]+p[2]\leq lp'[1]+p'[2]$ and $hp[1]+p[2]\leq hp'[1]+p'[2]$. That is $l(p[1]-p'[1])+(p[2]-p'[2])\leq 0$ and  $h(p[1]-p'[1])+(p[2]-p'[2])\leq 0$. Assume we have a linear function $f(t)=t(p[1]-p'[1])+(p[2]-p'[2])$ where $l< t< h$. Then we have $f(l)\leq 0$ and $f(h)\leq 0$. Because $f(t)$ is a line, therefore, for any value $t$ between the boundary values $l$ and $h$, we have $f(t)\leq 0$ for $l< t< h$.
\end{proof}

\begin{example}
Given $r\in [1/4,2]$, the dominance relationship is shown in Figure \ref{fig:eclipseQ}, we only need to determine if $S(p)\leq S(p')$ for $r=1/4$ and $r=2$ according to Theorem \ref{the:endpoint2D}. We take $p_2$ and $p_4$ as an example, $S(p_2)_{\langle 1/4\rangle}=1/4\times 4+4=5$. Similarly, we have $S(p_2)_{\langle 2\rangle}=12$, $S(p_4)_{\langle 1/4\rangle}=7$, and $S(p_4)_{\langle 2\rangle}=21$. Because $S(p_2)_{\langle 1/4\rangle}<S(p_4)_{\langle 1/4\rangle}$ and $S(p_2)_{\langle 2\rangle}<S(p_4)_{\langle 2\rangle}$, we have $p_2\prec_e p_4$.
\end{example}

Next, we show how to extend Theorem \ref{the:endpoint2D} from two dimensional space to high dimensional space.

\begin{theorem}\label{the:endpointhighD}
Given an attribute weight ratio vector $\textbf{r}=\langle r[1], ...,r[d-1]\rangle$, where $r[j]\in [l_j, h_j]$ for $j=1, 2, ..., d-1$, if $S(p)_\textbf{r}\leq S(p')_\textbf{r}$ for $r[j]=l_j$ and $r[j]=h_j$, where $j=1, 2, ...,d-1$, we have $S(p)_\textbf{r}\leq S(p')_\textbf{r}$ for all $r[j]\in [l_j,h_j]$.
\end{theorem}

\begin{proof}
Because the attribute weight ratio on each dimension $j$ can take the value of $l_j$ or $h_j$, we have $2^{d-1}$ domination vectors in $d$ dimensional space. Therefore, we have the following $2^{d-1}$ inequalities for $S(p)\leq S(p')$.

$\Sigma_{j=1}^{d-3}l_j(p[j]-p'[j])+\boldsymbol{l_{d-2}} (p[d-2]-p'[d-2])+\boldsymbol{l_{d-1}} (p[d-1]-p'[d-1])+(p[d]-p'[d])\leq 0$,

$\Sigma_{j=1}^{d-3}l_j(p[j]-p'[j])+\boldsymbol{l_{d-2}} (p[d-2]-p'[d-2])+\boldsymbol{h_{d-1}} (p[d-1]-p'[d-1])+(p[d]-p'[d])\leq 0$,

$\Sigma_{j=1}^{d-3}l_j(p[j]-p'[j])+\boldsymbol{h_{d-2}} (p[d-2]-p'[d-2])+\boldsymbol{l_{d-1}} (p[d-1]-p'[d-1])+(p[d]-p'[d])\leq 0$,

$\Sigma_{j=1}^{d-3}l_j(p[j]-p'[j])+\boldsymbol{h_{d-2}} (p[d-2]-p'[d-2])+\boldsymbol{h_{d-1}} (p[d-1]-p'[d-1])+(p[d]-p'[d])\leq 0$,

......,

$\Sigma_{j=1}^{d-1}h_j(p[j]-p'[j])+(p[d]-p'[d])\leq 0$.

Given the first two inequalities, according to Theorem \ref{the:endpoint2D}, we have

$\Sigma_{j=1}^{d-3}l_j(p[j]-p'[j])+\boldsymbol{l_{d-2}} (p[d-2]-p'[d-2])+r[d-1] (p[d-1]-p'[d-1])+(p[d]-p'[d])\leq 0$ (i), where $r[d-1]\in [l_{d-1},h_{d-1}]$. Similarly, given the third and fourth inequalities, we have

$\Sigma_{j=1}^{d-3}l_j(p[j]-p'[j])+\boldsymbol{h_{d-2}} (p[d-2]-p'[d-2])+r[d-1] (p[d-1]-p'[d-1])+(p[d]-p'[d])\leq 0$ (ii), where $r[d-1]\in [l_{d-1},h_{d-1}]$. Based on (i) and (ii), we have

$\Sigma_{j=1}^{d-3}l_j(p[j]-p'[j])+w[d-2] (p[d-2]-p'[d-2])+r[d-1] (p[d-1]-p'[d-1])+(p[d]-p'[d])\leq 0$, where $r[d-2]\in [l_{d-2},h_{d-2}]$ and $r[d-1]\in [l_{d-1},h_{d-1}]$. Similarly, we iteratively transform $l_j$ and $h_j$ to $r[j]$. Finally, we have

$\Sigma_{j=1}^{d-1}r[j](p[j]-p'[j])+(p[d]-p'[d])\leq 0$, where $r[j]\in [l_j,h_j]$, $j=1, 2, ..., d-1$.
\end{proof}

\begin{algorithm}[h] \scriptsize \caption{Baseline algorithm for computing eclipse points.}\label{Alg:baseline}
\SetKwInOut{Input}{input}\SetKwInOut{Output}{output}

\Input{a set of $n$ points in $d$ dimensional space.}
\Output{eclipse points.}

\For{i = 1 to n}{
    compute $S(p_i)_{\textbf{r}_k}, k=1,2,...,2^{d-1}$\;
    flag=1\;
    \For{j = 1 to n, $\neq i$}{
    compute $S(p_j)_{\textbf{r}_k}, k=1,2,...,2^{d-1}$\;
        \For{k= 1 to $2^{d-1}$}{
            \If{$S(p_j)_{\textbf{r}_k}>S(p_i)_{\textbf{r}_k}$}{
                goto Line 4\;}
        }
        flag=0\;
        break\;
    }
    \If{flag==1}{
        add $p_i$ to eclipse points\;
    }
}
\end{algorithm}

Based on Theorems \ref{the:endpoint2D} and \ref{the:endpointhighD}, the key idea of the baseline algorithm is that we only need to compare the scoring function $S(p_i)$ and $S(p_j)$ for each pair of points $p_i$ and $p_j$ with respect to all the domination vectors. The detailed algorithm is shown in Algorithm \ref{Alg:baseline}. We compute $S(p_i)_{\textbf{r}_k}$ corresponding to the $2^{d-1}$ domination vectors in Line 2 and set the flag as 1. We compute $S(p_j)_{\textbf{r}_k}$ corresponding to the $2^{d-1}$ domination vectors in Line 5. In Line 7, if $S(p_j)_{\textbf{r}_k}>S(p_i)_{\textbf{r}_k}$, it means that $p_j$ cannot eclipse-dominate $p_i$, and then we break from the forloop. If for all $k$ such that $S(p_j)_{\textbf{r}_k}\leq S(p_i)_{\textbf{r}_k}$, it means that $p_j\prec_e p_i$, and then we set the flag as 0. In Line 11, if the flag equals to $1$, it means that there is no other point that can eclipse-dominate $p_i$, so we add $p_i$ to eclipse points.

\begin{theorem}\label{the:3}
The time complexity of Algorithm \ref{Alg:baseline} is $O(n^22^{d-1})$.
\end{theorem}

\begin{proof}
Algorithm \ref{Alg:baseline} requires three forloops and each forloop iterates $n$, $n$, and $2^{d-1}$ times, respectively. Thus, Algorithm \ref{Alg:baseline} requires $O(n^22^{d-1})$ time in total.
\end{proof}

\subsection{Transformation-based Algorithm for Two Dimensional Space}\label{subsec:transformation-based2D}

In Subsection \ref{subsec:baseline}, we showed how to compute eclipse points in $O(n^2)$ time due to the two forloops. In this subsection, we show how to transform the eclipse problem to the skyline problem, and then we can employ an efficient $O(n\log n)$ time algorithm to solve the eclipse problem.

In the eclipse query, for each point $p_i(p_i[1],p_i[2])$, there are two domination lines with slopes $-h$ and $-l$, and $p_i$ eclipse-dominates the points in the domination range. For example, in Figure \ref{fig:transformation-based2D}, $p_1$ eclipse-dominates the points on the upper right of the two domination lines. For any two points $p$ and $p'$, the slopes ($-h$ and $-l$) of their two domination lines are the same. Therefore, if $p$ eclipse-dominates $p'$, the intercepts of $p's$ domination lines should be smaller than the corresponding intercepts of $p'$. Therefore, instead of directly comparing $S(p)$ and $S(p')$ for each pair of points for their eclipse-dominance which requires $O(n^2)$, we can map the intercepts of the two domination lines of each pair of points $p$ and $p'$ into a coordinate space, and compare their dominance utilizing $O(n\log n)$ skyline algorithm.

%\vspace{-1em}
\begin{figure}[htb]
 \centering
 \includegraphics[width=0.35\textwidth]{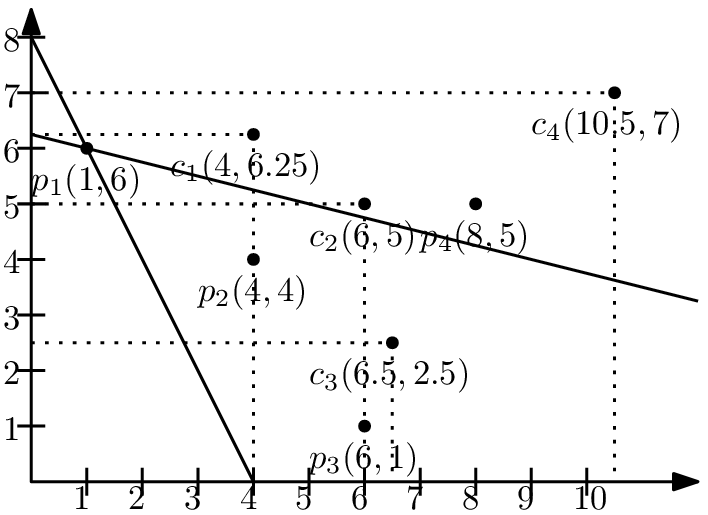}
  %\vspace{-1em}
 \caption{Mapping.}
 \label{fig:transformation-based2D}
\end{figure}
%\vspace{-1em}

Concretely, for each point $p_i$, the two domination lines have two intercepts with each dimension. We map $p_i$ into $c_i$ by taking the smaller intercept on the $j^{th}$ dimension as $c_i[j]$ because the larger intercept is already represented by the smaller intercept on the other dimension. For example, in Figure \ref{fig:transformation-based2D}, for point $p_1$, the two domination lines have two $y$-intercepts, $6.25$ and $8$. We take the smaller $y$-intercept, $6.25$, as $c_1[2]$. Similarly, we have $c_1[1]=4$. Therefore, we map $p_1$ into $c_1(4,6.25)$. Based on this mapping, we can prove that $p\prec_e p'$ if $c\prec_s c'$ in Theorem \ref{the:transform2D}.

\begin{algorithm}[h] \scriptsize \caption{Transformation-based algorithm for two dimensional space.}\label{Alg:transformation-based2D}
\SetKwInOut{Input}{input}\SetKwInOut{Output}{output}

\Input{a set of $n$ points in two dimensional space, an attribute weight ratio vector $\langle r\rangle$, where $r\in [l,h]$.}
\Output{eclipse points.}

\For{i =1 to n}{
    $c_i[1]=p_i[1]+p_i[2]/h$\;
    $c_i[2]=lp_i[1]+p_i[2]$\;
}

use $O(n\log n)$ algorithm to compute skyline points of $\{c_1, c_2,...,c_n\}$\;
add the corresponding points of skyline points to eclipse points\;
\end{algorithm}

\begin{theorem}\label{the:transform2D}
Given an attribute weight ratio $r\in [l,h]$ and a point $p_i$, we map $p_i$ into $c_i$, where $c_i[j]$ is the smaller intercept on the $j^{th}$ dimension of the two domination lines with slopes $-h$ and $-l$. We have $p\prec_e p'$ if $c\prec_s c'$ in two dimensional space. That is, the eclipse points of dataset $\{p_1, p_2, ..., p_n\}$ equal to the corresponding skyline points of dataset $\{c_1, c_2,...,c_n\}$.
\end{theorem}

\begin{proof}
For $p\prec_e p'$, we have
\begin{displaymath}
l  p[1]+p[2]\leq l  p'[1]+p'[2]
\end{displaymath}
\begin{displaymath}
h  p[1]+p[2]\leq h  p'[1]+p'[2]
\end{displaymath}
We map $p$ into $c$
\begin{displaymath}
c[2]=l  p[1]+p[2]
\end{displaymath}
\begin{displaymath}
c[1]=p[1]+\frac{1}{h}p[2]
\end{displaymath}
Similarly, we map $p'$ into $c'$
\begin{displaymath}
c'[2]=l  p'[1]+p'[2]
\end{displaymath}
\begin{displaymath}
c'[1]=p'[1]+\frac{1}{h}p'[2]
\end{displaymath}
Because $c\prec_s c'$, we have
\begin{displaymath}
l  p[1]+p[2]\leq l  p'[1]+p'[2]
\end{displaymath}
\begin{displaymath}
p[1]+\frac{1}{h}p[2]\leq p'[1]+\frac{1}{h}p'[2]
\end{displaymath}
It is easy to see that $p\prec_e p'$ can be equivalent to $c\prec_s c'$.
\end{proof}

Based on Theorem \ref{the:transform2D}, we show our algorithm for computing eclipse points in Algorithm \ref{Alg:transformation-based2D}. We compute $c_i$ of point $p_i$ in Lines 1-3. We employ the $O(n\log n)$ skyline algorithm to compute the skyline points of $\{c_1, c_2,..., c_n\}$ in Line 4.

\begin{example}
We show an example based on Figure \ref{fig:transformation-based2D}. Assume the attribute weight ratio $r\in [1/4,2]$. We have $c_1=(4,6.25)$, $c_2=(6,5)$, $c_3=(6.5,2.5)$, and $c_4(10.5,7)$. We compute the skyline points of $\{c_1,c_2,c_3,c_4\}$, and the skyline points are $c_1$, $c_2$, and $c_3$. Therefore, the corresponding eclipse points are $p_1$, $p_2$, and $p_3$.
\end{example}

\begin{theorem}\label{the:5}
The time complexity of Algorithm \ref{Alg:transformation-based2D} is $O(n\log n)$.
\end{theorem}

\begin{proof}
Lines 1-3 requires $O(n)$ time. Line 4 requires $O(n\log n)$ time. Thus, Algorithm \ref{Alg:transformation-based2D} requires $O(n\log n)$ time in total.
\end{proof}

\subsection{Transformation-based Algorithm for High Dimensional Space}\label{subsec:transformation-basedhighD}

In this subsection, we show how to transform the eclipse problem to the skyline problem for the high dimensional case similar to the two dimensional case by carefully choosing $d$ domination vectors, and then we can employ an efficient $O(n\log^{d-1} n)$ time algorithm to solve the eclipse problem.

For point $p(p[1],p[2],...,p[d])$, we have the following $2^{d-1}$ domination hyperplanes with respect to $2^{d-1}$ domination vectors as each dimension $j$ of the domination vector can take the boundary values $l_j$ and $h_j$.
\begin{displaymath}
l_1p[1]+l_2p[2]+...+l_{d-1}p[d-1]+p[d]=S(p)_{\textbf{r}_1};
\end{displaymath}
\begin{displaymath}
l_1p[1]+l_2p[2]+...+h_{d-1}p[d-1]+p[d]=S(p)_{\textbf{r}_1};
\end{displaymath}
......
\begin{displaymath}
h_1p[1]+h_2p[2]+...+l_{d-1}p[d-1]+p[d]=S(p)_{\textbf{r}_{2^{d-1}-1}};
\end{displaymath}
\begin{displaymath}
h_1p[1]+h_2p[2]+...+h_{d-1}p[d-1]+p[d]=S(p)_{\textbf{r}_{2^{d-1}}};
\end{displaymath}
We can write the domination vector function matrix as follows.
\[
\begin{bmatrix}
    l_1 & l_2 & ... & l_{d-1} & 1 \\
    l_1 & l_2 & ... & h_{d-1} & 1 \\
    ......\\
    h_1 & h_2 & ... & l_{d-1} & 1 \\
    h_1 & h_2 & ... & h_{d-1} & 1 \\
\end{bmatrix}
\]
Because there are $d$ variables, the rank of the domination vector function matrix is at most $d$. Therefore, we can carefully choose $d$ domination vectors from these $2^{d-1}$ domination vectors to represent the original matrix. Furthermore, the rank of the new $d$-row matrix should be $d$.

We can choose the first row to identify the $d^{th}$ attribute weight ratio and the row with $r[k]=h_j, k=j$ and $r[k]=l_j, j\neq k,k=1,2,...,d-1$ to identify the $j^{th}$ attribute weight ratio. It is easy to see that there are $d-1$ such rows. Therefore, we get a new $d$-row matrix with rank $d$. In fact, each row in the new matrix corresponds to a $c_i[j]$. For example, if we choose $r[1]=h_1$ and $r[j]=l_j$ for $j=2,...,d-1$, the corresponding domination vector is $\textbf{v}_1=\langle h_1,l_2,...,l_{d-1},1 \rangle$. This domination vector $\textbf{v}_1$ corresponds to $c_i[1]$ because we get the smallest $x$-intercept by the domination hyperplane determined by $\textbf{v}_1$. Given the above mapping, we can show that $p\prec_e p'$ if $c\prec_s c'$ in high dimensional space as in the following theorem.

\begin{theorem}\label{the:transformHD}
Given an attribute weight ratio vector $\textbf{r}=\langle r[1],r[2], ...,r[d-1]\rangle$ and a point $p_i$, we map $p_i$ into $c_i$, where $c_i[j]$ is the smallest intercept on the $j^{th}$ dimension of the $d$ domination hyperplanes. We have $p\prec_e p'$ if $c\prec_s c'$ in high dimensional space. That is, the eclipse points of dataset $\{p_1, p_2,...,p_n\}$ equal to the corresponding skyline points of dataset $\{c_i, c_2,...,c_n\}$.
\end{theorem}

\begin{proof}
If $p\prec_e p'$, for all $r[j]\in \{l_j,h_j\}$, where $j=1,2,...,d-1$, we have $S(p)\leq S(p')$ for $2^{d-1}$ domination vectors. We can carefully choose $d$ domination vectors to represent these $2^{d-1}$ domination vectors as follows.
\[
\begin{bmatrix}
    l_1 & l_2 & ... & l_{d-1} & 1 \\
    l_1 & l_2 & ... & \boldsymbol{h_{d-1}} & 1 \\
    ......\\
    l_1 & \boldsymbol{h_2} & ... & l_{d-1} & 1 \\
    \boldsymbol{h_1} & l_2 & ... & l_{d-1} & 1 \\
  \end{bmatrix}
\]
We have $S(p)\leq S(p')$ for these $d$ domination vectors as follows.
\begin{displaymath}
  \sum_{j=1}^{d-1}l_jp[j]+p[d]\leq  \sum_{j=1}^{d-1}l_jp'[j]+p'[d]
\end{displaymath}

and for $j=1,2,...,d-1$
\begin{equation}\label{equ:rank}
  h_jp[j]+\sum_{k=1,k\neq j}^{d-1}l_kp[k]+p[d]\leq h_jp'[j] \sum_{k=1,k\neq j}^{d-1}l_kp'[k]+p'[d]
\end{equation}

If $c\prec _sc'$, for all $j=1,2,...,d$, we have $c[j]\leq c'[j]$. That is
\begin{displaymath}
\frac{p[d]+h_1p[1]+\sum_{k=2,k\neq 1}^{d}l_{k}p[k]}{h_1}
\end{displaymath}
\begin{displaymath}
\leq \frac{p'[d]+h_1p'[1]+\sum_{k=2,k\neq 1}^{d}l_{k}p'[k]}{h_1}
\end{displaymath}
\begin{displaymath}
\frac{p[d]+h_2p[2]+\sum_{k=2,k\neq 2}^{d}l_{k}p[k]}{h_2}
\end{displaymath}
\begin{displaymath}
\leq \frac{p'[d]+h_2p'[2]+\sum_{k=2,k\neq 2}^{d}l_{k}p'[k]}{h_2}
\end{displaymath}
\begin{displaymath}
  ......
\end{displaymath}
\begin{displaymath}
\frac{p[d]+h_{d-1}p[d-1]+\sum_{k=2,k\neq d-1}^{d}l_{k}p[k]}{h_{d-1}}
\end{displaymath}
\begin{displaymath}
\leq \frac{p'[d]+h_{d-1}p'[d-1]+\sum_{k=2,k\neq d-1}^{d}l_{k}p'[k]}{h_{d-1}}
\end{displaymath}
\begin{displaymath}
l_1p[1]+l_2p[2]+...+l_{d-1}p[d-1]+p[d]
\end{displaymath}
\begin{displaymath}
\leq l_1p'[1]+l_2p'[2]+...+l_{d-1}p'[d-1]+p'[d];
\end{displaymath}
which is equivalent to Equation \ref{equ:rank}.
\end{proof}

Based on Theorem \ref{the:transformHD}, the detailed algorithm for computing eclipse points in high dimensional space is shown in Algorithm \ref{Alg:transformation-basedHigherD}. For each point $p_i$, we map it into $c_i$ in Lines 1-4, and then use the $O(nlog^{d-1}n)$ high dimensional skyline algorithm to compute skyline points of $\{c_1,c_2,...,c_n\}$. The corresponding points $p_i,1\leq i\leq n$ are eclipse points.

\begin{algorithm}[h] \scriptsize \caption{Transformation-based algorithm for high dimensional space.}\label{Alg:transformation-basedHigherD}
\SetKwInOut{Input}{input}\SetKwInOut{Output}{output}

\Input{a set of $n$ points in high dimensional space, attribute weight ratio vector $\textbf{r}=\langle r[1], ..., r[d-1]\rangle$, where $r[j]\in [l_j, h_j]$ for $j=1,2,...,d-1$.}
\Output{eclipse points.}

\For{i =1 to n}{

    $c_i[d]=\sum_{j=1}^{d-1}l_jp_i[j]+p_i[d]$\;
    \For{j = 1 to d-1}{
        $c_i[j]=\frac{p_i[d]+h_jp_i[j]+\sum_{k=2,k\neq j}^{d-1}l_{k}p_i[k]}{h_j}$\;
    }
}

use the $O(n\log^{d-1} n)$ ECDF algorithm \cite{DBLP:journals/cacm/Bentley80} to compute skyline points of $\{c_1,c_2,...,c_n\}$\;
add the corresponding points of skyline points to Eclipse points\;
\end{algorithm}

\begin{theorem}\label{the:7}
The time complexity of Algorithm \ref{Alg:transformation-basedHigherD} is $O(n\log^{d-1} n)$.
\end{theorem}
\vspace{-1em}
\begin{proof}
Lines 1-4 require $O(nd)$ time. Line 5 requires $O(n\log^{d-1} n)$ time. Thus, Algorithm \ref{Alg:transformation-basedHigherD} requires $O(n\log^{d-1} n)$ time in total.
\end{proof}

\section{Index-based Algorithms}\label{sec:index-based}
The transformation-based algorithm we presented in Section \ref{sec:transformation-based} is more efficient than the baseline algorithm. However it is still computationally expensive for real time queries, since it computes each query for the entire dataset from scratch. In this section, we show more efficient algorithms utilizing index structures and duality transform.

For a better perspective, we transform our problem from the primal space to the dual space by duality transform \cite{de2000computational}. For a point $p=(p[1],p[2],...,p[d])$, its dual hyperplane is $x_d=p[1]x_1+p[2]x_2+...+p[d-1]x_{d-1}-p_d$. For a hyperplane $x_d=p[1]x_1+p[2]x_2+...+p[d-1]x_{d-1}+p[d]$, its dual point is $(p[1],p[2],...,p[d-1],-p[d])$. For example, in Figure \ref{fig:dual2D}, for point $p_1(1,6)$, its corresponding line in the dual space is $y=x-6$. We show how to construct the index structures and process the queries in two dimensional space in Subsection \ref{subsec:index-based2D} and high dimensional space in Subsection \ref{subsec:index-basedhigh}.

\subsection{Index-based Algorithm for Two Dimensional Space}\label{subsec:index-based2D}

We first show how to compute $1$NN and skyline in the dual space in Figure \ref{fig:dual2D}, and then introduce our algorithm for computing eclipse points. We have four points $p_1,p_2,p_3,p_4$ in the primal space (Figure \ref{fig:dual2D}(a)) and their corresponding lines in the dual space (Figure \ref{fig:dual2D}(b)). $p_4$ cannot be an eclipse point because it is skyline-dominated by $p_2$ and $p_3$, thus, we omit $p_4$ in the dual space.

\begin{figure}[htb]
 \centering
 \includegraphics[width=0.4\textwidth]{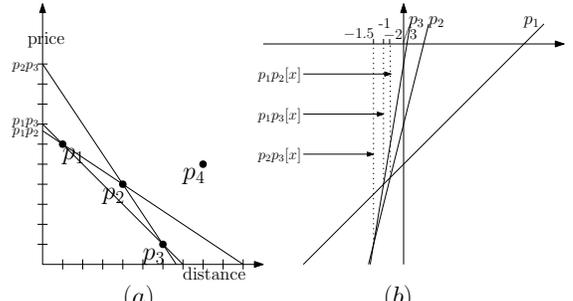}
  \vspace{-1em}
 \caption{(a): Primal space, (b) Dual space.}
 \label{fig:dual2D}
\end{figure}

For all 1NN, skyline, and eclipse queries, we need to find the point that is not in the domination range of any other point. But the domination range differs for each of the queries. For $1$NN, the domination range of each point given attribute weight ratio $r=l$ is determined by the domination line with slope $-l$. Correspondingly, in the dual space, given the $x$-coordinate ($-l$), we need to find the line that is not dominated by any other line, i.e., the closest line to the $x$-axis. For example, in Figure \ref{fig:dual2D}, if $l=2$, the nearest neighbor is $p_1$ in the primal space. Correspondingly, in the dual space, line $p_1$ is the closest line to the $x$-axis when $x=-2$.

For skyline, the domination range of each point is determined by the two domination lines with slopes $0$ and $\infty$, i.e., the attribute weight ratio $r\in (-\infty, 0]$. Correspondingly, in the dual space, given the $x$-coordinate range $(-\infty,0]$, we need to find the lines that are not dominated by any other line. We say line $p_a$ dominates $p_b$ if $p_a$ is consistently closer to the $x$-axis than $p_b$ for the entire range. For example, in Figure \ref{fig:dual2D}, in interval $(-\infty,p_2p_3[x]]$ of the $x$-axis, the closest line to the $x$-axis is $p_1$, and in interval $({p_1p_3}[x],0]$ of the $x$-axis, the closest line to the $x$-axis is $p_3$, where ${p_ip_j}[x]$ is the $x$-coordinate of the intersection of lines $p_i$ and $p_j$ in the dual space. However, in the entire query range, there is no line that can dominate $p_2$. Therefore, the skyline query result is $p_1$, $p_2$, and $p_3$.

For eclipse, given a ratio range $r\in [l,h]$, the domination range of each point is determined by the two domination lines with slopes $-h$ and $-l$. Correspondingly, in the dual space, given the $x$-coordinate range $[-h,-l]$, we need to find the lines that are not dominated by any other line, i.e., consistently closer to the $x$-axis within the range. Since the order of the lines (in their closeness to the $x$-axis) only changes when the two lines intersect, in order to quickly find the lines that are not dominated by any other line within the query range $[-h,-l]$, we can partition the $x$-axis by intervals where for each interval, the order of the lines does not change. The partitioning points are naturally determined by the intersections between each pair of lines. This motivates us to build 1) an index structure (Order Vector Index) that stores the intervals and their corresponding order of the lines, and 2) an index structure (Intersection Index) that stores the intersections which affect the consistent order of the lines within a range. In this way, given any query range, we can quickly retrieve all the lines that are not dominated by any other line within the range.

\subsubsection{Indexing}
In this subsection, we show how to build Order Vector Index and Intersection Index.

\partitle{Order Vector Index}
We build an Order Vector Index which partitions the $x$-axis into ${u\choose 2}+1$ intervals, and each entry $ov_i$ corresponding to an interval stores the order of the lines in their closeness to the $x$-axis, where $u$ is the number of skyline points and $u\choose 2$ is the number of intersections between the $u$ lines in the dual space. For example, in Figure \ref{fig:indexing}, we partition the $x$-axis into four intervals. The last interval is $(-2/3,0]$, and it stores the order of the lines $\textbf{ov}_4=\langle 2,1,0 \rangle$ corresponding to $p_3,p_2,p_1$.

\begin{algorithm}[h] \scriptsize \caption{Indexing in two dimensional space.}\label{Alg:index-basedPre2D}
\SetKwInOut{Input}{input}\SetKwInOut{Output}{output}

\Input{a set of $n$ points in two dimensional space.}
\Output{Order Vector Index and Intersection Index.}

use the $O(n\log n)$ time complexity algorithm to compute the $u$ skyline points of $n$ points\;

\For{i = 1 to u}{
    compute the corresponding line $p_i$ in the dual space for point $p_i$ \;
}

\For{i =1 to u-1}{
    \For{j=i+1 to u}{
        compute value $p_ip_j[x]$ by eliminating $y$ from two lines $p_i$ and $p_j$\;}}

sort those $p_ip_j[x]$ in ascending order $v_1,v_2,...,v_{u\choose 2}$, and record the corresponding two lines for each $v_{i}$, $i=1,2,...,{u\choose 2}$\;

\For{i=2 to ${u\choose 2} +1$}{

\For{j =1 to u}{
    $startY_j=(v_{i-1}+\epsilon)p_j[1]-p_j[2]$ ($v_1-\epsilon$ if $i=1$)\;
}
\For{k =1 to u}{
    compute $\textbf{ov}_{i-1}[k]$\;
}
}
\end{algorithm}

The detailed algorithm is shown in Algorithm \ref{Alg:index-basedPre2D}. Because eclipse points are the subset of skyline points, we find skyline points in Line 1. In Lines 2-3, we compute the corresponding line $p_i$ in the dual space for point $p_i$. In Lines 4-5, we compute value $p_ip_j[x]$ by eliminating $y$ from two lines $p_i$ and $p_j$, where $p_ip_j[x]$ is the $x$-coordinate of the intersection of lines $p_i$ and $p_j$. In Line 7, we sort all $u\choose 2$ intersections' $x$-coordinates. Because the order of the lines in their closeness to the $x$-axis does not change in the same interval, we choose $v_i+\epsilon$ as the $x$-coordinate to compute the $y$-coordinate in Line 10 as $startY_i$ which is used to determine the order of the lines in their closeness to the $x$-axis in the interval, where $\epsilon$ is a very small number. The reason for adding a very small number is that we compute the initial order in the interval rather than on the interval boundary $v_i$. Finally, we obtain the $\textbf{ov}_i$ for each interval in Line 12. $\textbf{ov}_i[k]$ means there are $\text{ov}_i[k]$ lines that can dominate line $p_k$ to the $x$-axis in the $i^{th}$ interval.

\partitle{Intersection Index}
For two dimensional case, we record the corresponding two lines for each interval boundary $v_i$ in Line 7 in Algorithm \ref{Alg:index-basedPre2D}, and the interval boundaries in Order Vector Index are exactly the $x$-coordinates of the intersections. Therefore, we can use Order Vector Index to index both order vectors and intersections in two dimensional space.

\begin{figure}[htb]
 \centering
 \includegraphics[width=0.45\textwidth]{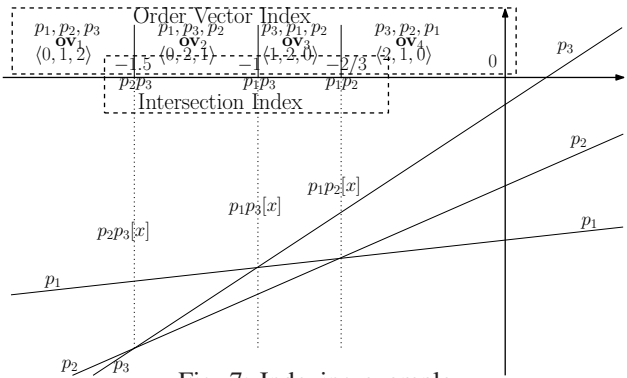}
  \vspace{-1em}
 \caption{Indexing example.}
 \label{fig:indexing}
\end{figure}
\vspace{-1em}

\begin{example}
We show how to build the index structures in Figures \ref{fig:dual2D} and \ref{fig:indexing}. In Line 3, we map points $p_1(1,6)$, $p_2(4,4)$, $p_3(6,1)$ in the primal space (shown in Figure \ref{fig:dual2D}(a)) into linear equations $y=x-6$, $y=4 x-4$, $y=6 x-1$ in the dual space (shown in Figure \ref{fig:dual2D}(b)) by duality transform, respectively. In Line 6, taking two lines $p_1$ and $p_2$ as an example, we substitute $y=4 x-4$ into $y=x-6$, and get the $x$-coordinate of the intersection of lines $p_1$ and $p_2$, $p_1p_2[x]=-2/3$. Similarly, we have $p_1p_3[x]=-1$ and $p_2p_3[x]=-1.5$. In Line 7, we sort $-1.5,-1,-2/3$ and record the corresponding two lines for each value, e.g., $-2/3$ is formed by two lines $p_1$ and $p_2$. Therefore, we have four intervals $(-\infty,-1.5]$, $(-1.5,-1]$, $(-1,-2/3]$, and $(-2/3,0]$. Taking the last interval $(-2/3,0]$ as an example, we compute its corresponding $\textbf{ov}_4$ in Lines 9-12. We compute the $y$-coordinate $startY_i$ for lines $p_1$, $p_2$, and $p_3$ in Line 10. We have $startY_1=-1/2-6=-6.5$ by choosing $\epsilon=1/6$. Similarly, we have $startY_2=-6$, and $startY_3=-4$. Therefore, the corresponding $\textbf{ov}_4$ is $ \langle 2,1,0\rangle$.
\end{example}

\subsubsection{Querying}

In this subsection, we show how to process the query. Given an attribute weight ratio $r\in [l,h]$, the corresponding query range in the dual space is $[-h,-l]$. Each interval stores the order of the lines in their closeness to the x-axis within that interval.  Our goal is to determine if a line is consistently closer than other lines within the query range $[-h,-l]$.  Hence we can start with the order in one of the intervals as the initial order, then determine if the order is consistent by enumerating through the intersection points in the range.

The detailed algorithm is shown in Algorithm \ref{Alg:index-basedQuery2D}. We get the initial $\textbf{ov}$ for $-l$ from Order Vector Index and use binary search to find the intersections whose $x$-coordinates lie between $-h$ and $-l$ from Intersection Index in Lines 1-2. We assume there are $m$ intersections within the query range. The worst case for $m$ is $u\choose 2$, i.e., all the intersections lie in the query range, but $m$ is much smaller than $u\choose 2$ in practice given the small query range. For the intersection of lines $p_a$ and $p_b$, if $\textbf{ov}[a]$ is smaller than $\textbf{ov}[b]$, that means line $p_a$ dominates $p_b$ before this intersection (in $[-h,p_ap_b[x]]$). After this intersection, i.e., in $[p_ap_b[x],-l]$, $p_b$ dominates $p_a$. That is, the number of lines that can dominate $p_b$ will be subtracted by one (Line 6). Similarly, if $\textbf{ov}[a]$ is larger than $\textbf{ov}[b]$, we have $\textbf{ov}[a]--$ in Line 8. Finally, if $\textbf{ov}[i]=0$, that is there is no other point that can dominate $p_i$. We add $p_i$ to eclipse points.
\vspace{-1em}

\begin{algorithm}[h] \scriptsize \caption{Query for two dimensions.}\label{Alg:index-basedQuery2D}
\SetKwInOut{Input}{input}\SetKwInOut{Output}{output}

\Input{Order Vector Index and Intersection Index, attribute weight ratio $r\in [l,h]$.}
\Output{eclipse points.}

use binary search to find the interval that contains $-l$ and get the initial $\textbf{ov}$ for $-l$ from Order Vector Index;

find these intersections whose $x$-coordinate lying between $-h$ and $-l$ from Intersection Index\;

\For{i = 1 to m (number of intersections)}{
    assume lines $p_a$ and $p_b$ form the intersection\;
    \If{$\textbf{ov}[a]<\textbf{ov}[b]$}{
        $\textbf{ov}[b]--$;
    }
    \Else
    {
        $\textbf{ov}[a]--$;
    }
}

\For{i =1 to u}{
    \If{$\textbf{ov}[i]=0$}{
        add $p_i$ to eclipse points\;
    }
}

\end{algorithm}

\begin{example}
Given $r\in [1/4,2]$, we have the corresponding query $[-2,-1/4]$. We search $-l=-1/4$ which belongs to the interval of $(-2/3,0]$ from Order Vector Index, and get the initial $\textbf{ov}_4=\langle 2,1,0\rangle$ in Line 1. In Line 2, we find intersections $p_1p_2$, $p_1p_3$, $p_2p_3$ from Intersection Index because their $x$-coordinates lie in $[-2,-1/4]$. After intersection $p_1p_2$, $p_2$ cannot dominate $p_1$, thus the number of the lines that can dominate $p_1$ should be subtracted by 1. That is $\textbf{ov}_4[1]=2-1=1$. Similarly, after intersection $p_1p_3$, $\textbf{ov}_4[1]=1-1=0$, after intersection $p_2p_3$, $\textbf{ov}_4[2]=1-1=0$. Finally, we get $\textbf{ov}_4=\langle 0,0,0\rangle $. That is, there is no other line that can dominate $p_1, p_2, p_3$ in $[-2,-1/4]$. Thus, we add points $p_1, p_2, p_3$ to eclipse points.
\end{example}

%\vspace{-1em}
\begin{table}[htb]\centering\footnotesize
\caption{Example of query for $r\in [1/4,2]$.}\label{tab:table}
 \vspace{-1em}
{%
\begin{tabular}{|c||c|c|c|}
\hline
                         & $p_1$ & $p_2$ & $p_3$\\
\hline
$startY_i$ & -6.5    & -6    & -4 \\
\hline
initial $\textbf{ov}_4$         & 2 & 1 & 0\\
\hline
after $p_1p_2$ $\textbf{ov}_4$  & 1 & 1 & 0\\
\hline
after $p_1p_3$ $\textbf{ov}_4$  & 0 & 1 & 0\\
\hline
after $p_2p_3$ $\textbf{ov}_4$  & 0 & 0 & 0\\
\hline
\end{tabular}}
\end{table}%
\vspace{-1em}

\begin{theorem}\label{the:8}
The time complexity of Algorithm \ref{Alg:index-basedQuery2D} is $O(u+m)$.
\end{theorem}

\begin{proof}
Line 1 requires $O(\log u^2)$ time. Line 2 requires $O(\log u^2+m)$ time. Lines 3-8 can be finished in $O(m)$ time. Lines 9-11 can be finished in $O(u)$ time. Thus, Algorithm \ref{Alg:index-basedQuery2D} requires $O(u+m)$ time in total.
\end{proof}

\subsection{Index-based Algorithm for High Dimensional Space}\label{subsec:index-basedhigh}

In this subsection, we show how to build the index structures and process the eclipse query in high dimensional space. The general idea is very similar to two dimensional case. In the dual space of two dimensional space, we need to find the lines that are not dominated by any other line with respect to the $x$-axis (line $y=0$) within the query range $x\in [-h,-l]$. Similarly, in the dual space of high dimensional space, we need to find the hyperplanes that are not dominated by any other hyperplane with respect to the hyperplane $x_d=0$ within the query range $x_1\in [-h_1,-l_1]$,...,$x_{d-1}\in [-h_{d-1},-l_{d-1}]$ for $d$ dimensional space. Therefore, we need Order Vector Index to index the initial order of those hyperplanes in each hypercell with respect to the hyperplane $x_d$=0. In two dimensional space, we want to find the intersections whose $x$-coordinates lying in $[-h,-l]$. Similarly, in high dimensional space, we want to find the $d-1$ dimensional intersecting hyperplanes that are intersecting with range $x_1\in [-h_1,-l_1]$,...,$x_{d-1}\in [-h_{d-1},-l_{d-1}]$ for $d$ dimensional space. Therefore, we need Intersection Index to index the intersecting $d-1$ dimensional hyperplanes for any two $d$ dimensional hyperplanes.

\subsubsection{Indexing}
In this subsection, we show how to build Order Vector Index and Intersection Index.

\partitle{Order Vector Index}
We show how to build Order Vector Index in high dimensional space in Algorithm \ref{Alg:index-basedPreHigherD}. We first compute skyline points using $O(n\log ^{d-1}n)$ time skyline algorithm in Line 1. For each skyline point, we compute its corresponding hyperplane $p_i$ for point $p_i$ in $d$ dimensional dual space in Lines 2-3. In Line 6, we compute the intersecting $d-1$ dimensional hyperplane of $d$ dimensional hyperplanes $p_i$ and $p_j$ by eliminating $x_d$. For these $u \choose 2$ intersecting $d-1$ dimensional hyperplanes, we compute the arrangement\footnote{Let $L$ be a set of $n$ lines in the plane. The set $L$ induces a subdivision of the plane that consists of vertices, edges, and faces. This subdivision is usually referred to as the arrangement induced by $L$ \cite{de2000computational}. Similarly, in high dimensional space, let $H$ be a set of $n$ hyperplanes in $d$ dimensional space. The set $H$ induces a subdivision of the space that consists of vertices, edges, faces, facets, and hypercells.} in Line 7. In Line 8, for each hypercell in the arrangement, we compute its initial $\textbf{ov}_i$ which records the order of the hyperplanes in their closeness to hyperplane $x_d=0$. Therefore, in the query phase, we only need to locate any point from the query range to get the corresponding hypercell and then get the corresponding $\textbf{ov}$ of the hypercell in logarithmic time.

\begin{algorithm}[h] \scriptsize \caption{Indexing in high dimensional space.}\label{Alg:index-basedPreHigherD}
\SetKwInOut{Input}{input}\SetKwInOut{Output}{output}

\Input{a set of $n$ points in high dimensional space.}
\Output{Order Vector Index.}

use the $O(n\log^{d-1} n)$ time complexity algorithm to compute the $u$ skyline points of $n$ points\;
\For{i = 1 to u}{
    compute the dual $d$ dimensional hyperplane of $p_i$\;

}

\For{i = 1 to u-1}{
    \For{j = i+1 to u}{
        compute the $d-1$ dimensional hyperplane of hyperplanes $dh_i$ and $dh_j$ by eliminating $x_d$\;
    }
}

compute the arrangement of these $u\choose 2$ hyperplanes in $d-1$ dimensional space\;
for each hypercell, compute its initial $\textbf{ov}$\;

\end{algorithm}

\partitle{Intersection Index}
It is easy to see that the dominating part of the query for high dimensional space is to find these pairs of hyperplanes whose intersecting hyperplanes intersect with the query range. If we scan all the $u\choose 2$ intersecting $d-1$ dimensional hyperplanes to determine if they intersect with the $d-1$ dimensional query range, the time cost is prohibitively high. Therefore, we show how to index these $u\choose 2$ intersecting hyperplanes to facilitate the search of intersecting hyperplanes in Intersection Index.

We first show a Line Quadtree\footnote{We call line quadtree in two dimensional space or hyperplane octree in high dimensional space, but for the sake of simplicity, we use the term ``line quadtree'' to refer to both two and high dimensional space.} with good average case performance, then show a Cutting Tree with good worst case performance. We note that we are indexing lines/hyperplanes, so the traditional indexes, e.g., R-tree, are not suitable. We use three dimensional space as an example, which corresponds to finding the intersecting lines that are intersecting with rectangle $x_1\in [-h_1,-l_1]$, $x_2\in [-h_2,-l_2]$.

(\emph{Line Quadtree})
Line quadtree is a rooted tree in which every internal node has four children in two dimensional space. In genral, a $d$ dimensional hyperplane octree has $2^d$ children in each internal node. Every node in line quadtree corresponds to a square. If a node $t$ has children, then their corresponding squares are the four quadrants of the square of $t$, referred to as NE, NW, SW, and SE. Figure \ref{fig:quadtree} illustrates an example of line quadtree and the corresponding subdivision. We set the maximum capacity for each node as $3$, i.e., we need to partition a square into four subdivisions if there are more than $3$ lines going through this square.

\begin{figure}[htb]
 \centering
 \includegraphics[width=0.4\textwidth]{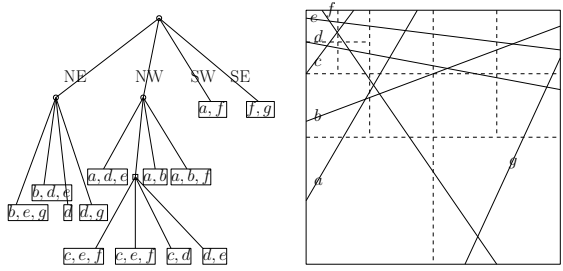}
 \caption{Line quadtree with maximum capacity $3$.}
 \label{fig:quadtree}
\end{figure}

It is easy to see how to construct line quadtree. The recursive definition of line quadtree is immediately translated into a recursive algorithm: split the current square into four quadrants, partition the line set accordingly, and recursively construct line quadtrees for each quadrant with its associated line set. The recursion stops when the line set contains less than maximum capacity lines. Line quadtree can be constructed in $O(de\times n)$ time with $O(de\times n)$ nodes, where $de$ is the depth of the corresponding line quadtree.

To query line quadtree is straightforward. We start at the root node and examine each child node to check if it intersects the range being queried for. If it does, recurse into that child node. Whenever we encounter a leaf node, examine each entry to see if it intersects with the query range, and return it if it does. Finally, we combine all the returned lines.

Line quadtree has a very good performance in the average case, however, the worst case is $O(n)$, i.e., the depth for line quadtree is $O(n)$ in the worst case. Therefore, we present an alternative index structure cutting tree which has a good worst case guarantee ($O(\log n)$ time complexity).

\vspace{-1em}
\begin{figure}[htb]
 \centering
 \includegraphics[width=0.4\textwidth]{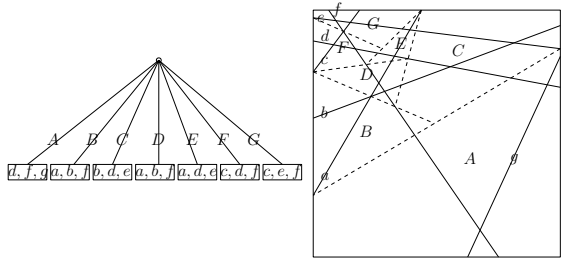}
 \caption{A $(1/3)$-cutting of size seven for a set of seven lines.}
 \label{fig:cuttingtree}
\end{figure}

(\emph{Cutting Tree})
Cutting tree partitions the space into a set of possibly unbounded triangles with the property that no triangle is crossed by more than $n/t$ lines, where $n$ is the number of lines and $t\in [1,n]$ is a parameter. For any set $L$ of $n$ lines in the plane, a $(1/t)$-cutting of size $O(t^2)$ exists. Moreover, such a cutting can be constructed in $O(nt)$ time \cite{de2000computational}. We show a $(1/3)$-cutting of size seven for a set of $7$ lines in Figure \ref{fig:cuttingtree}. For each triangle, there are at most ${7/3}$ intersecting lines. We note that cutting tree also can be designed as a tree structure as line quadtree.

\begin{theorem}\cite{de2000computational}
For any parameter $t$, it is possible to compute a $(1/t)$-cutting of size $O(t^d)$ with a deterministic algorithm that takes $O(nt^{d-1})$ time. The resulting data structure has $O(\log ^dn)$ query time. The query time can be reduced to $O(\log n)$.
\end{theorem}

We note that the deterministic algorithms for constructing cutting tree based on the arrangement structure are theoretical in nature and involve large constant factors \cite{de2000computational}. Therefore, in this paper, we implement the cutting tree index structure using the probabilistic schemes \cite{DBLP:conf/stoc/Clarkson85}\cite{DBLP:conf/stoc/Clarkson86}, which will be shown in the experimental section. We randomly choose $t$ hyperplanes from $n$ hyperplanes, and the formed arrangement structure will have a high probability satisfying the requirement of cutting tree. However, constructing arrangement is also a prohibitively high cost task with time complexity $O(n^d)$ \cite{DBLP:journals/siamcomp/EdelsbrunnerOS86}. Following the same spirit, the space with more hyperplanes will be chosen. For the space with more hyperplanes, the space will be chosen with a higher probability to be partitioned. Due to the same problem of constructing arrangement, we cannot process the point location in logarithmic time in high dimensional space in practice. Therefore, we compute $\textbf{ov}$ in Line 1 of Algorithm $\ref{Alg:index-basedQueryHigherD}$, which requires $O(u)$ time and does not impact the entire time complexity.

\subsubsection{Querying}

In this subsection, we show how to process the eclipse query based on Order Vector Index and Intersection Index in Algorithm \ref{Alg:index-basedQueryHigherD}. The idea is very similar to the two dimensional case. In Line 1, we get the initial $\textbf{ov}$ by point location in $O(\log n)$ time. We then quickly find the intersecting $d-1$ dimensional hyperplanes based on Intersection Index in Line 2. For any $d-1$ dimensional hyperplane formed by $d$ dimensional hyperplanes $p_a$ and $p_b$, if $\textbf{ov}[a]<\textbf{ov}[b]$, that means before this intersecting $d-1$ dimensional hyperplane, $p_a$ dominates $p_b$, thus, we have $\textbf{ov}[a]--$. Otherwise, we have $\textbf{ov}[b]--$. Finally, if $\textbf{ov}[i]=0$, we add $p_i$ to eclipse points.

\begin{algorithm}[h] \scriptsize \caption{Query for high dimensions.}\label{Alg:index-basedQueryHigherD}
\SetKwInOut{Input}{input}\SetKwInOut{Output}{output}

\Input{Order Vector Index and Intersection Index, attribute weight ratio vector $\langle r[1], r[2],...,r[d-1]\rangle$, where $r[j]\in [l_j, h_j]$ for $j=1,2,...,d-1$.}
\Output{eclipse points.}

choose any point from the query range to process the point location to get the initial $\textbf{ov}$ from Order Vector Index\;

find these $d-1$ dimensional hyperplanes intersecting with the query range $x_1\in [-h_1,-l_1], x_2\in [-h_2,-l_2]$,...,$x_{d-1}\in [-h_{d-1},-l_{d-1}]$ from Intersection Index\;

\For{i = 1 to m (number of intersecting $d-1$ dimensional hyperplanes)}{
    assume hyperplanes $p_{a}$ and $p_b$ form this $d-1$ dimensional hyperplane\;
    \If{$\textbf{ov}[a]<\textbf{ov}[b]$}{
        $\textbf{ov}[b]--$;
    }
    \Else
    {
        $\textbf{ov}[a]--$;
    }
}

\For{i =1 to u}{
    \If{$\textbf{ov}[i]=0$}{
        add $p_i$ to eclipse points\;
    }
}

\end{algorithm}

\begin{theorem}\label{the:10}
The time complexity of Algorithm \ref{Alg:index-basedQueryHigherD} is $O(u+m)$.
\end{theorem}

\begin{proof}
Lines 1 can be finished in $O(\log u^2)$ time based on the arrangement structure \cite{DBLP:journals/comgeo/ChazelleF94}. Line 2 requires $O(\log u^2+m)$ time based on the cutting tree index structure. Lines 3-8 require $O(m)$ time. Lines 9-11 require $O(u)$ time. Thus, Algorithm \ref{Alg:index-basedQueryHigherD} requires $O(u+m)$ time in total.
\end{proof}

\section{Experiments}\label{sec:Experiments}
In this section, we present experimental studies evaluating our algorithms for computing eclipse points.

\subsection{Experiment Setup}
We implemented the following algorithms in Python and ran experiments on a machine with Intel Core i7 running Ubuntu with 8GB memory.
\vspace{-1em}

\begin{itemize}
\item \textbf{BASE:} Baseline algorithm (Section \ref{sec:transformation-based}).\vspace{-0.7em}

\item \textbf{TRAN:} Transformation-based Algorithm (Section \ref{sec:transformation-based}).\vspace{-0.7em}

\item \textbf{QUAD:} Index-based Algorithm with Line QuadTree (Section \ref{sec:index-based}).\vspace{-0.7em}

\item \textbf{CUTTING:} Index-based Algorithm with Cutting Tree (Section \ref{sec:index-based}).
\end{itemize}
\vspace{-1em}

We used both synthetic datasets and a real NBA dataset in our experiments. To study the scalability of our methods, we generated independent (INDE), correlated (CORR), and anti-correlated (ANTI) datasets following the seminal work \cite{DBLP:conf/icde/BorzsonyiKS01}. We also built a real dataset that contains 2384 NBA players. The data was extracted from http://stats.nba.com/leaders/alltime/?ls=iref:nba:gnav on 04/15/2015. Each player has five attributes that measure the player's performance. These attributes are Points (PTS), Rebounds (REB), Assists (AST), Steals (STL), and Blocks (BLK). The parameter settings are shown in Table \ref{tab:par}. For the sake of simplicity and w.l.o.g., we consider $r[1]=r[2]=...=r[d-1]$ in this paper.

\vspace{-1em}
\begin{table}[htb]\centering
\caption{Parameter settings (defaults are in bold).}\label{tab:par}
\vspace{-1em}
{%
\footnotesize
\begin{tabular}{|c|c|}
\hline
Param. & Settings\\
\hline
$n$ & $2^7$, $\mathbf{2^{10}}$, $2^{13}$, $2^{17}$, $2^{20}$\\
\hline
$d$ & $2$, $\textbf{3}$, $4$, $5$\\
\hline
$r[j]$ & $[0.18,5.67]$, $\mathbf{[0.36,2.75]}$, $[0.58,1.73]$, $[0.84,1.19]$\\
angle       & $[100,170]$, ~~$\mathbf{[110,160]}$, ~~$[120,150]$, ~$[130,140]$\\
\hline
\end{tabular}}
\end{table}%

\partitle{Cutting Tree Implementation}
Given a dataset of $n$ points in $d$ dimensional space, we map these $n$ points into $n$ hyperplanes in $d$ dimensional dual space correspondingly. For any two $d$ dimensional hyperplanes, they form a $d-1$ dimensional hyperplane by eliminating $x_d$. Therefore, we have $n \choose 2$ intersecting hyperplanes in $d-1$ dimensional space. For any $d-1$ dimensional hyperplanes in $d-1$ dimensional space, they form a $d-1$ dimensional point. Therefore, we have ${n\choose 2} \choose {d-1}$ intersecting points in $d-1$ dimensional space.  For these ${n\choose 2} \choose {d-1}$ intersecting points, we randomly choose $t^d$ points. We employ the classic Voronoi algorithm to compute the Voronoi hypercells for these $t^d$ points, i.e., the $d-1$ dimensional space is partitioned into $t^d$ regions. The intuition for this implementation is clear. For a subregion with more hyperplanes, there has more intersecting points. If we randomly sample points, then this region will have more points to be selected with high probability. With more points to be selected, this subregion can be partitioned into more subregions, which leads to the smaller number of hyperplanes intersecting with each subregion. We note that this implementation is better than the cutting tree implementation based on the arrangement structure even from the theoretical perspective because the time complexity for computing Voronoi diagram is $O(n^{d/2})$ \cite{DBLP:conf/focs/Chazelle91} while computing arrangement structure requires $O(n^d)$ time \cite{DBLP:journals/siamcomp/EdelsbrunnerOS86} in $d$ dimensional space.

\subsection{Case Study}\label{subsec:caseStudy}
We performed a user study using the hotel example (Figure \ref{fig:1nn}). We posted a questionnaire using the conference scenario to ask $38$ students and staff members in our department and $30$ workers from Amazon Mechanical Turk. We asked them to choose the best hotel reservation system. The hotel reservation systems include skyline system, top-$k$ system, eclipse-ratio system, e.g., $r[1]\in [0.3,0.5]$, eclipse-weight system, e.g., $w[1]\in [0.3,0.5]$ and $w[2]=1-w[1]$, and eclipse-category system, e.g., $w[1]$ is very important/important/similar/unimportant/very unimportant compared to $w[2]$, where each category corresponds to a range. We received $61$ responses in total.

Table \ref{tab:casestudy} shows the number of answers for each hotel systems. The results show that eclipse-category system attracts more attentions and our algorithms for computing eclipse points can be easily adapted for each of the eclipse systems.

\begin{table}[htb]\centering
\caption{Results of case study.}\label{tab:casestudy}
 \vspace{-1em}
{%
\begin{tabular}{|c|c|c|c|c|c|}
\hline
skyline & top-$k$ & eclipse-ratio & eclipse-weight & eclipse-category\\
\hline
13 & 7 & 8 & 8 & 25\\
\hline
\end{tabular}}
\end{table}%

\subsection{Average Number of Eclipse Points}\label{subsec:average}
In this subsection, we study the average number of eclipse points on the independent and identically distributed datasets, which can be used for designing attribute weight ratio vector. It is easy for a user to set the attribute weight vector, but it is hard for the user to estimate how many eclipse points will be returned. If we compute the expected number of eclipse points in advance, the user can adjust the attribute weight ratio vector according to the desired number of eclipse points.

The number of eclipse points for different number of points $n$, different number of dimensions $d$, and different attribute weight ratio vectors $\textbf{r}=\langle r[1],r[2],...,r[d-1]\rangle$ are shown in Tables \ref{tab:diffn}, \ref{tab:diffd}, and \ref{tab:diffAttri}, respectively. We can see that the number of points has very small impact on the number of eclipse points, but the number of dimensions and the attribute weight ratios have significant impact.

%\vspace{-1em}
\begin{table}[htb]\centering
\caption{Expected number of eclipse points vs. $n$.}\label{tab:diffn}
 \vspace{-1em}
{%
\begin{tabular}{|c|c|c|c|c|c|c|c|}
\hline
$n$ & $2^7$ & $2^{10}$ & $2^{13}$ & $2^{17}$ & $2^{20}$\\
\hline
\# eclipse points & 3.71 & 3.83 & 3.91 & 4.03 & 4.13\\
\hline
\end{tabular}}
\end{table}%
\vspace{-1em}

\begin{table}[htb]\centering
\caption{Expected number of eclipse points vs. $d$.}\label{tab:diffd}
 \vspace{-1em}
{%
\begin{tabular}{|c|c|c|c|c|c|c|c|}
\hline
$d$ & $2$ & $3$ & $4$ & $5$\\
\hline
\# eclipse points & $1.8$ & $3.8$ & $8.5$ & $17.2$\\
\hline
\end{tabular}}
\end{table}%

\begin{table}[htb]\centering
\caption{Expected number of eclipse points vs. $r$.}\label{tab:diffAttri}
 \vspace{-1em}
{%
\begin{tabular}{|c|c|c|c|c|c|c|c|}
\hline
$r$ & [0.18,5.67] & [0.36,2.75] & [0.58,1.73] & [0.84,1.19]\\
\hline
\# ecl. pts & 7.2 & 3.8 & 2.2 & 1.3\\
\hline
\end{tabular}}
\end{table}%

\begin{figure*}[!htb]
\centering
\subfigure[\scriptsize{time cost of CORR}]{
\begin{minipage}[b]{0.22\textwidth}
\includegraphics[width=1.1\linewidth]{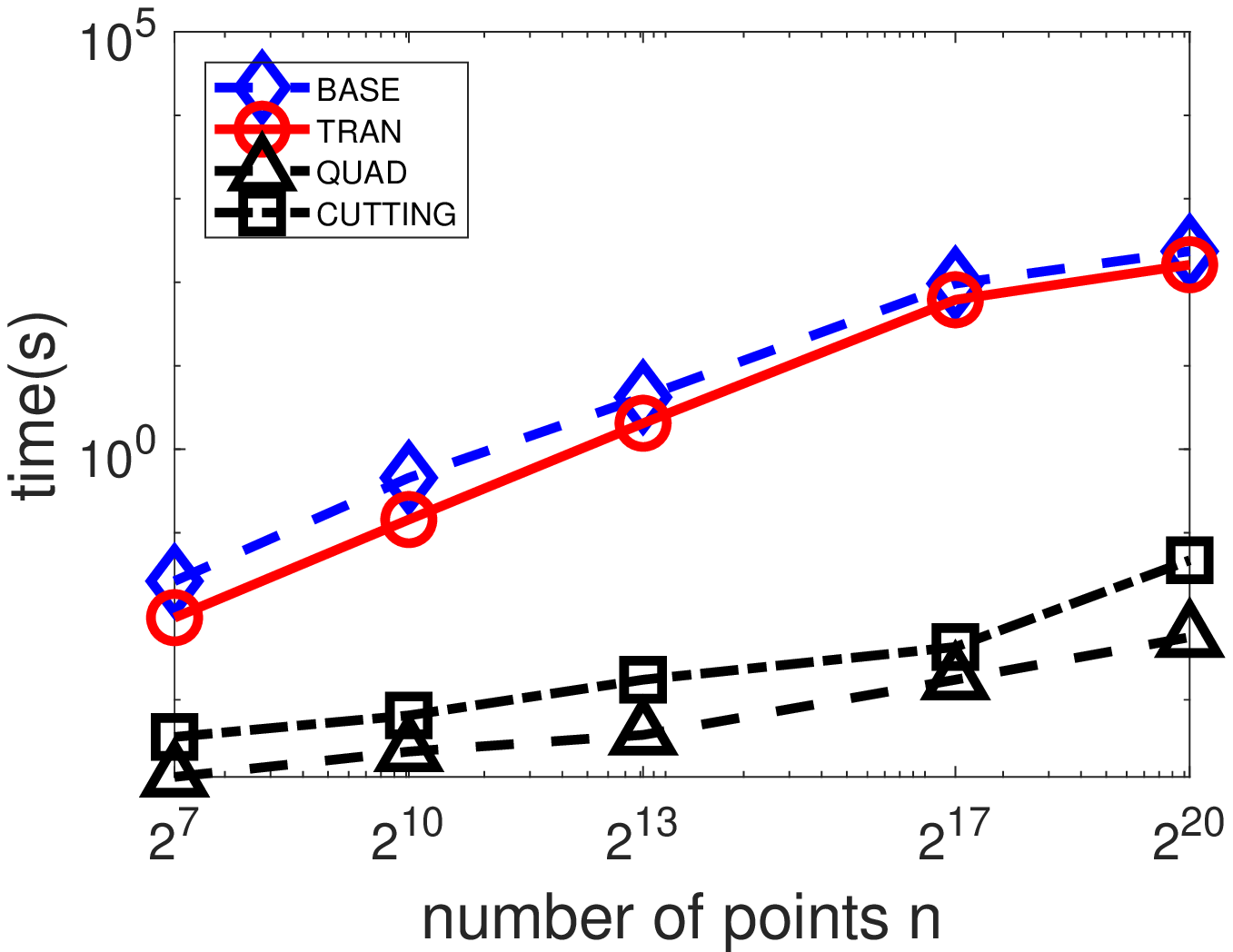}
\end{minipage}
}
\subfigure[\scriptsize{time cost of INDE}]{
\begin{minipage}[b]{0.22\textwidth}
\includegraphics[width=1.1\linewidth]{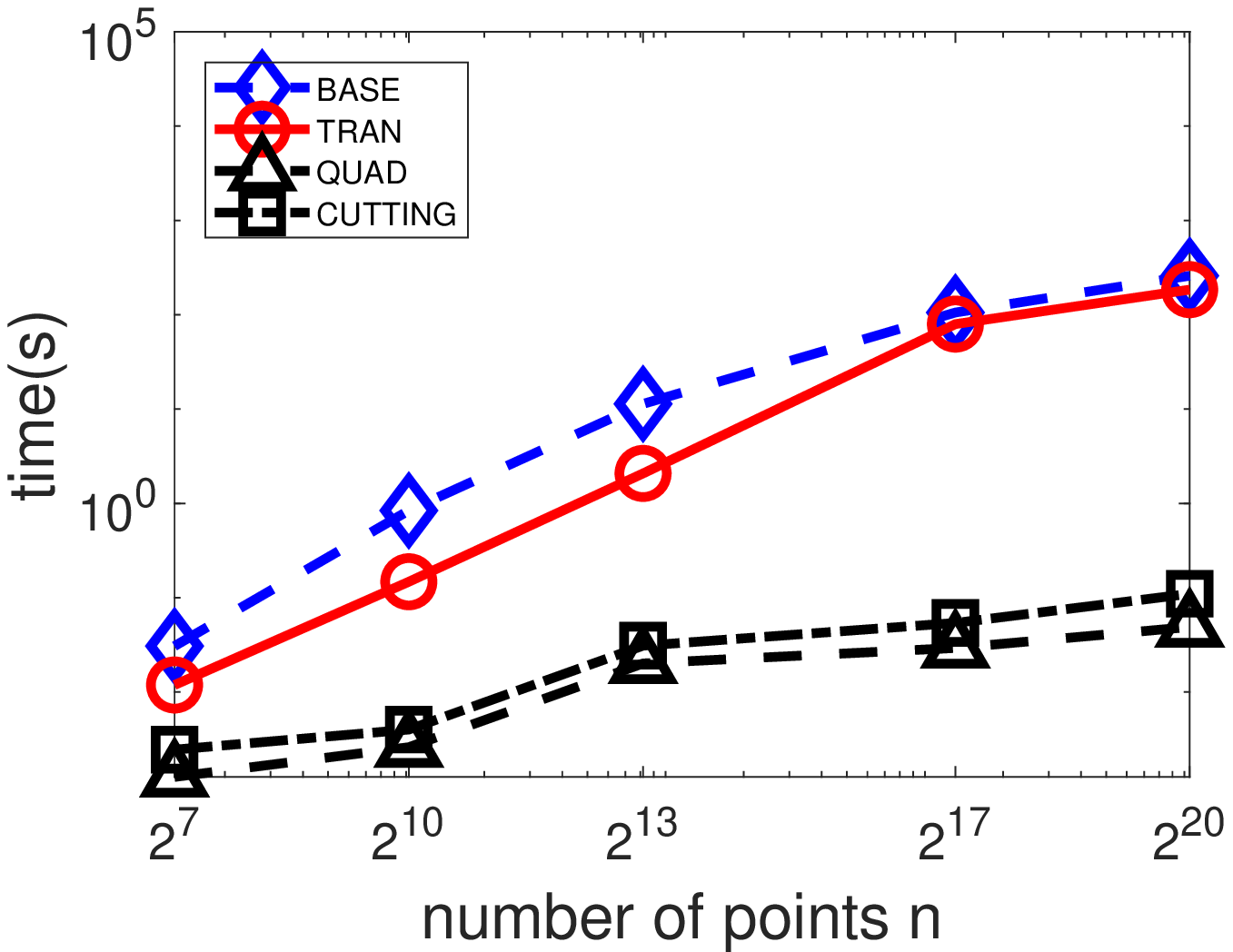}
\end{minipage}
}
\subfigure[\scriptsize{time cost of ANTI}]{
\begin{minipage}[b]{0.22\textwidth}
\includegraphics[width=1.1\linewidth]{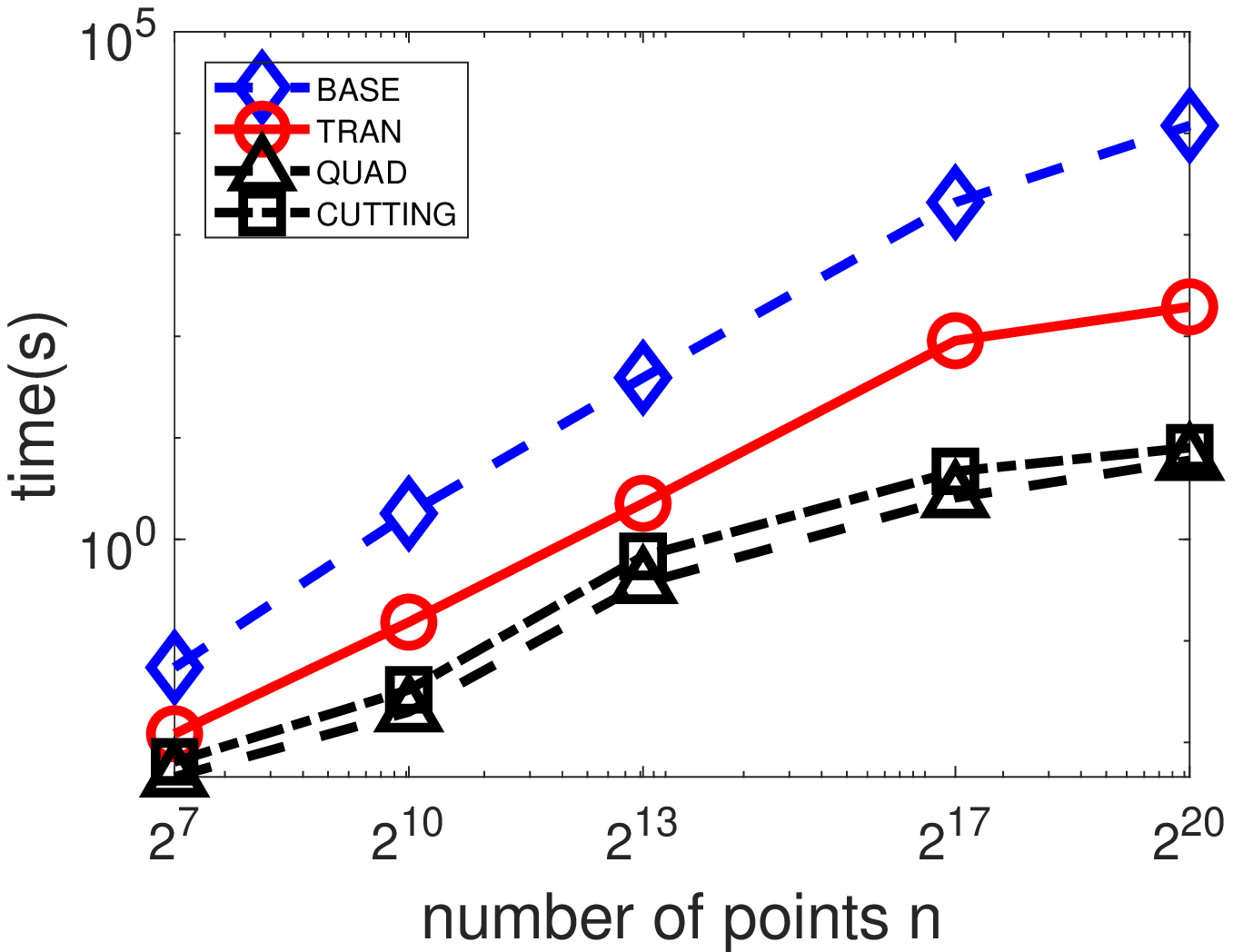}
\end{minipage}
}
\subfigure[\scriptsize{time cost of NBA}]{
\begin{minipage}[b]{0.22\textwidth}
\includegraphics[width=1.1\linewidth]{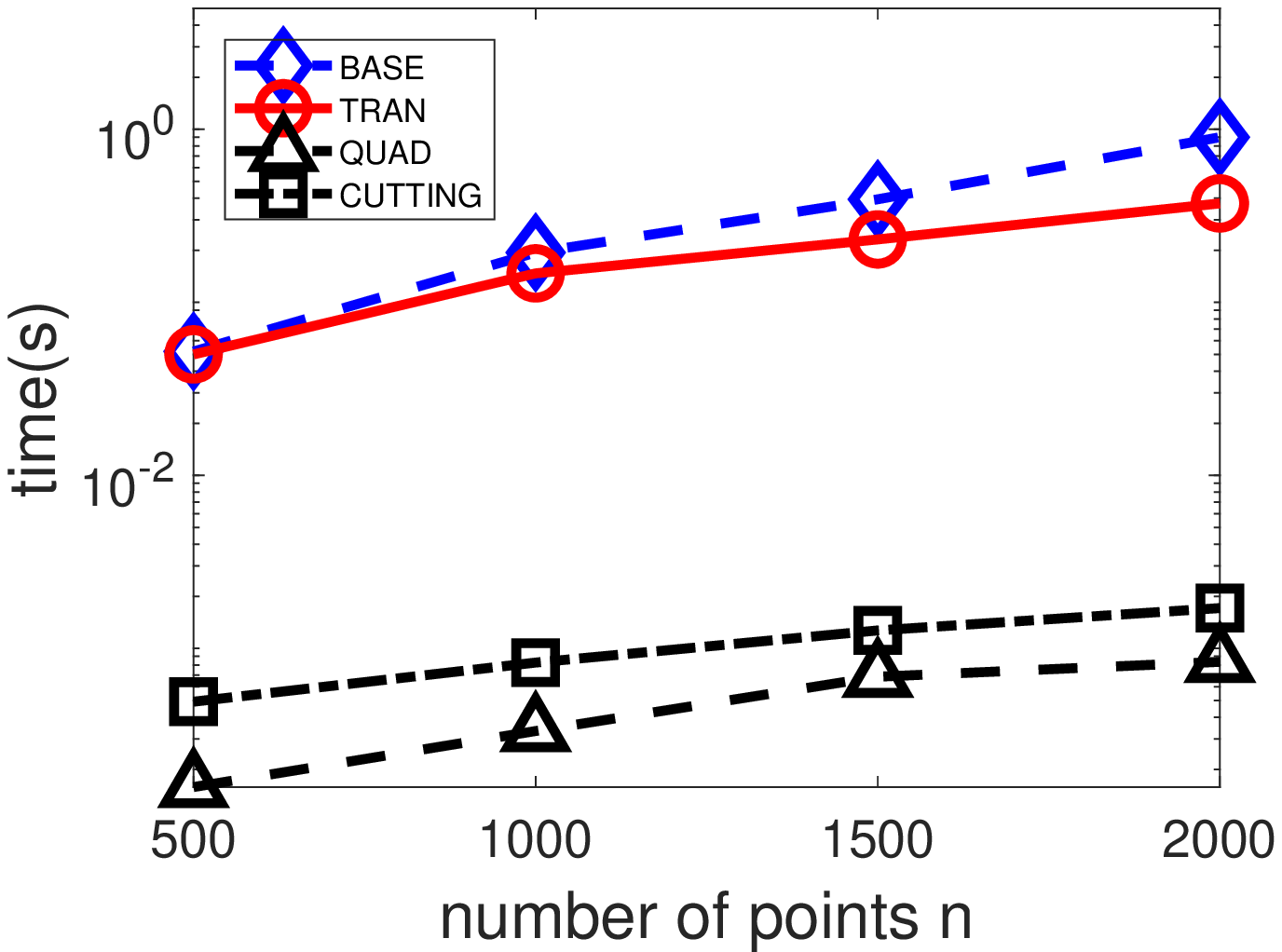}
\end{minipage}
}
 \vspace{-1em}\captionsetup{font={scriptsize}}\caption{The impact of $n$.} \label{fig:nAverage}
\end{figure*}

\begin{figure*}[!htb]
\centering
\subfigure[\scriptsize{time cost of CORR}]{
\begin{minipage}[b]{0.22\textwidth}
\includegraphics[width=1.1\linewidth]{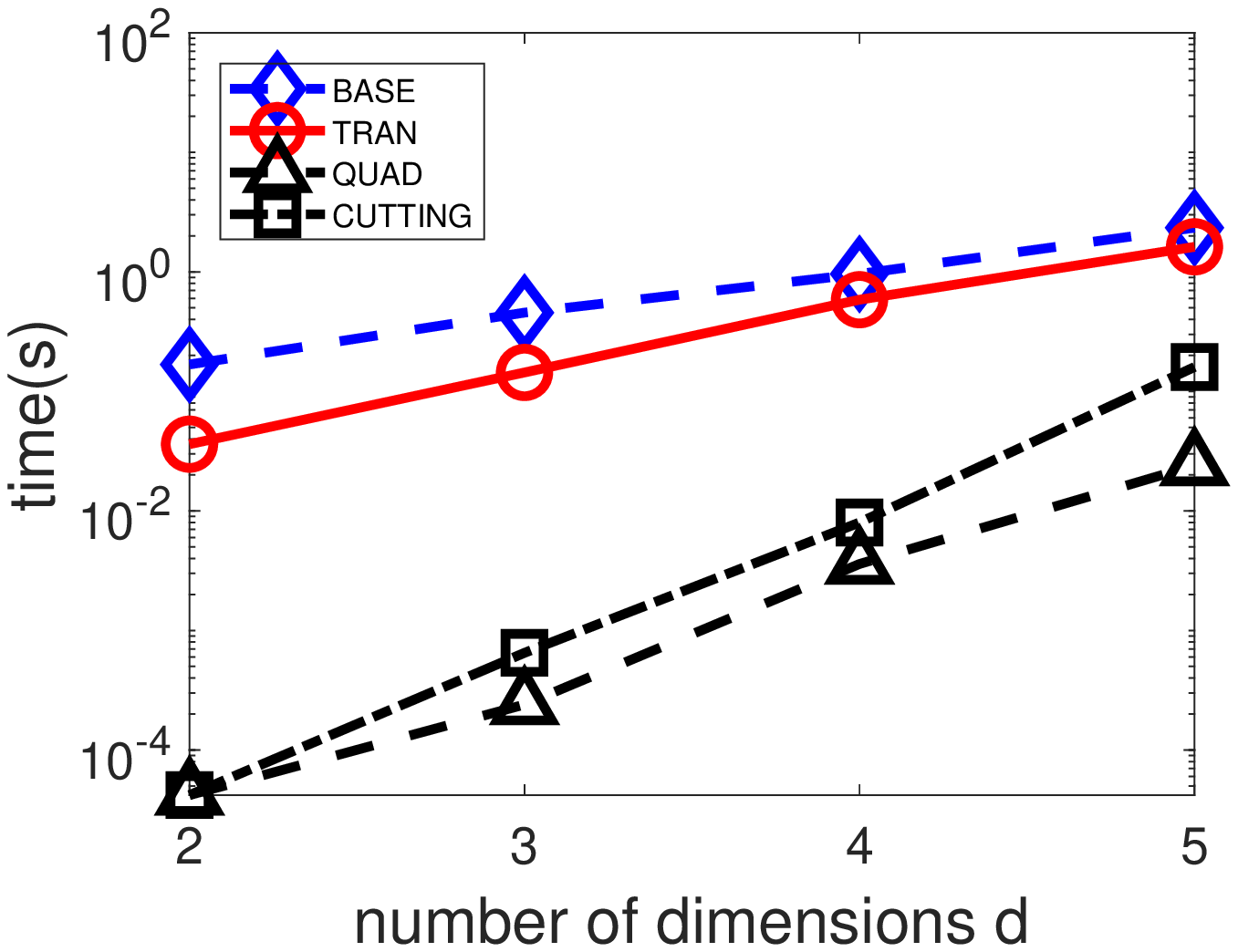}
\end{minipage}
}
\subfigure[\scriptsize{time cost of INDE}]{
\begin{minipage}[b]{0.22\textwidth}
\includegraphics[width=1.1\linewidth]{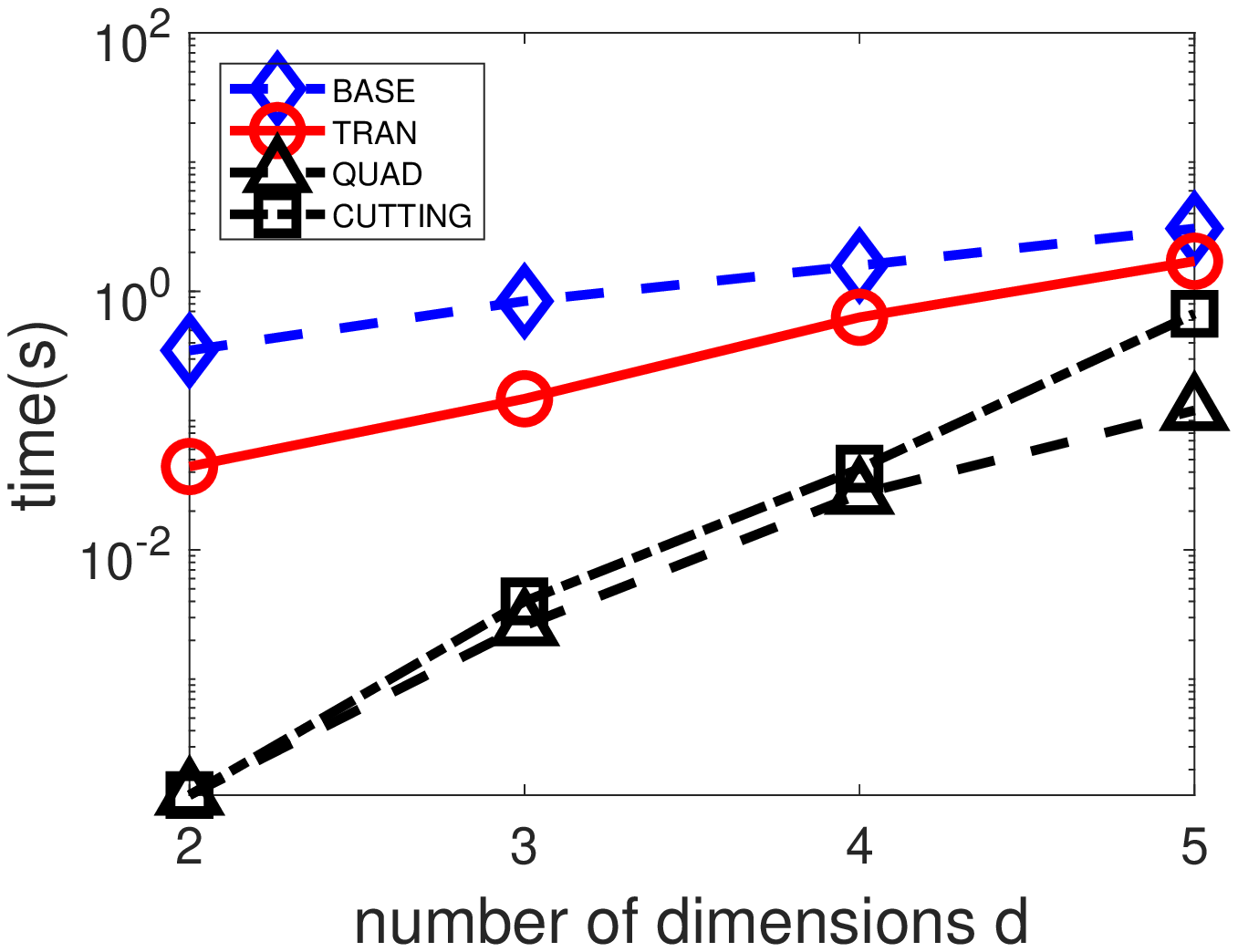}
\end{minipage}
}
\subfigure[\scriptsize{time cost of ANTI}]{
\begin{minipage}[b]{0.22\textwidth}
\includegraphics[width=1.1\linewidth]{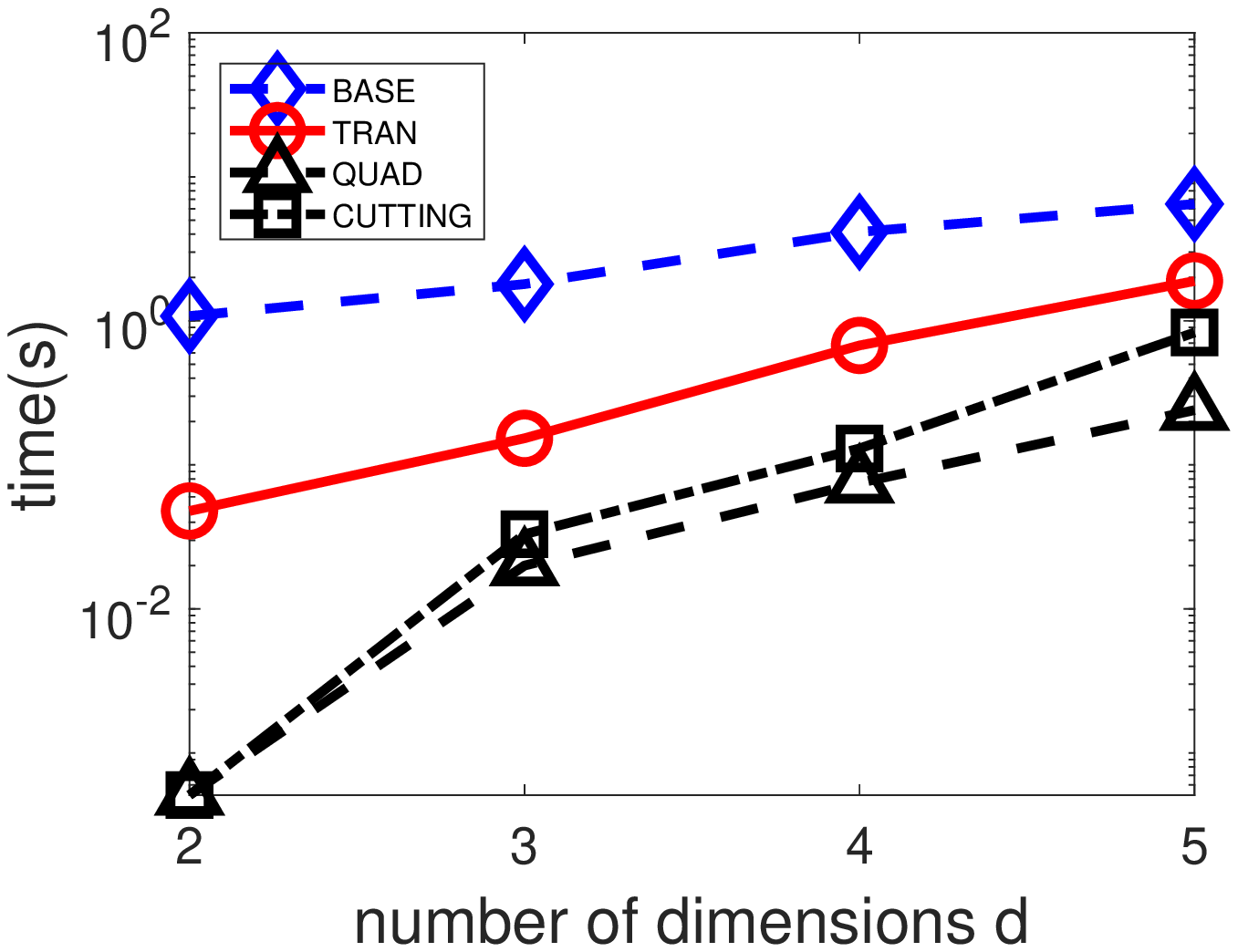}
\end{minipage}
}
\subfigure[\scriptsize{time cost of NBA}]{
\begin{minipage}[b]{0.22\textwidth}
\includegraphics[width=1.1\linewidth]{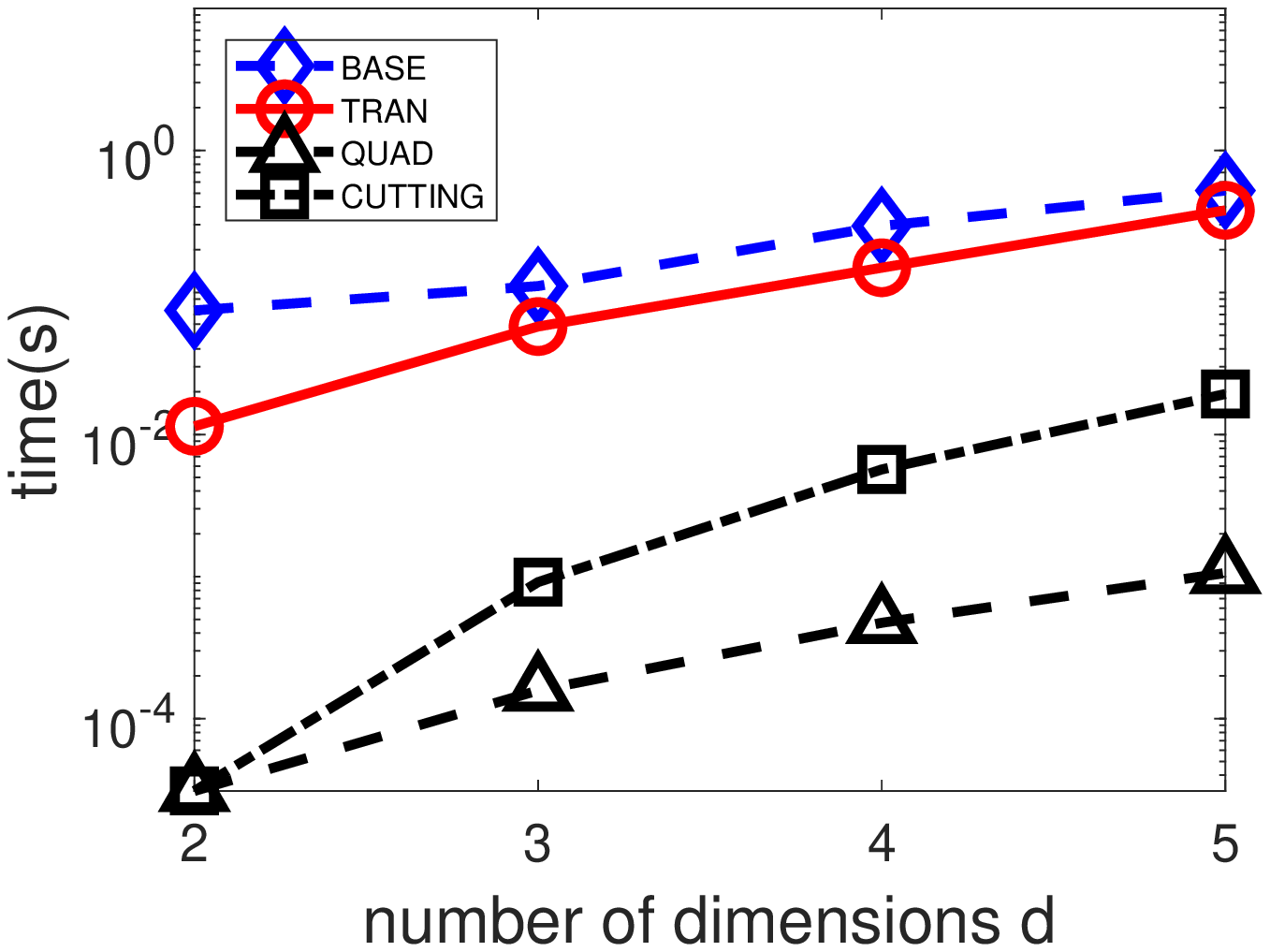}
\end{minipage}
}
 \vspace{-1em}\captionsetup{font={scriptsize}}\caption{The impact of $d$.} \label{fig:dAverage}
\end{figure*}

\begin{figure*}[!htb]
\centering
\subfigure[\scriptsize{time cost of CORR}]{
\begin{minipage}[b]{0.22\textwidth}
\includegraphics[width=1.1\linewidth]{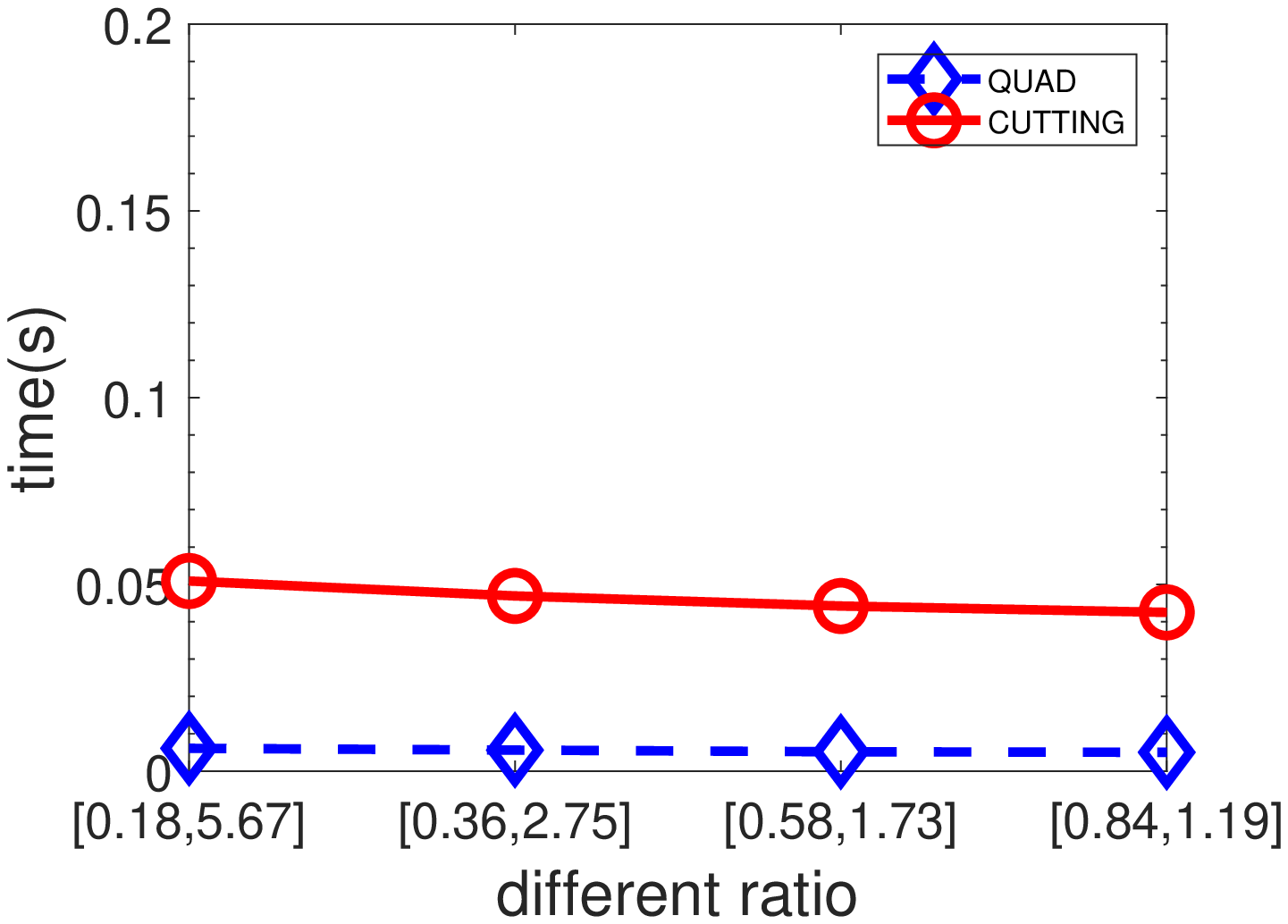}
\end{minipage}
}
\subfigure[\scriptsize{time cost of INDE}]{
\begin{minipage}[b]{0.22\textwidth}
\includegraphics[width=1.1\linewidth]{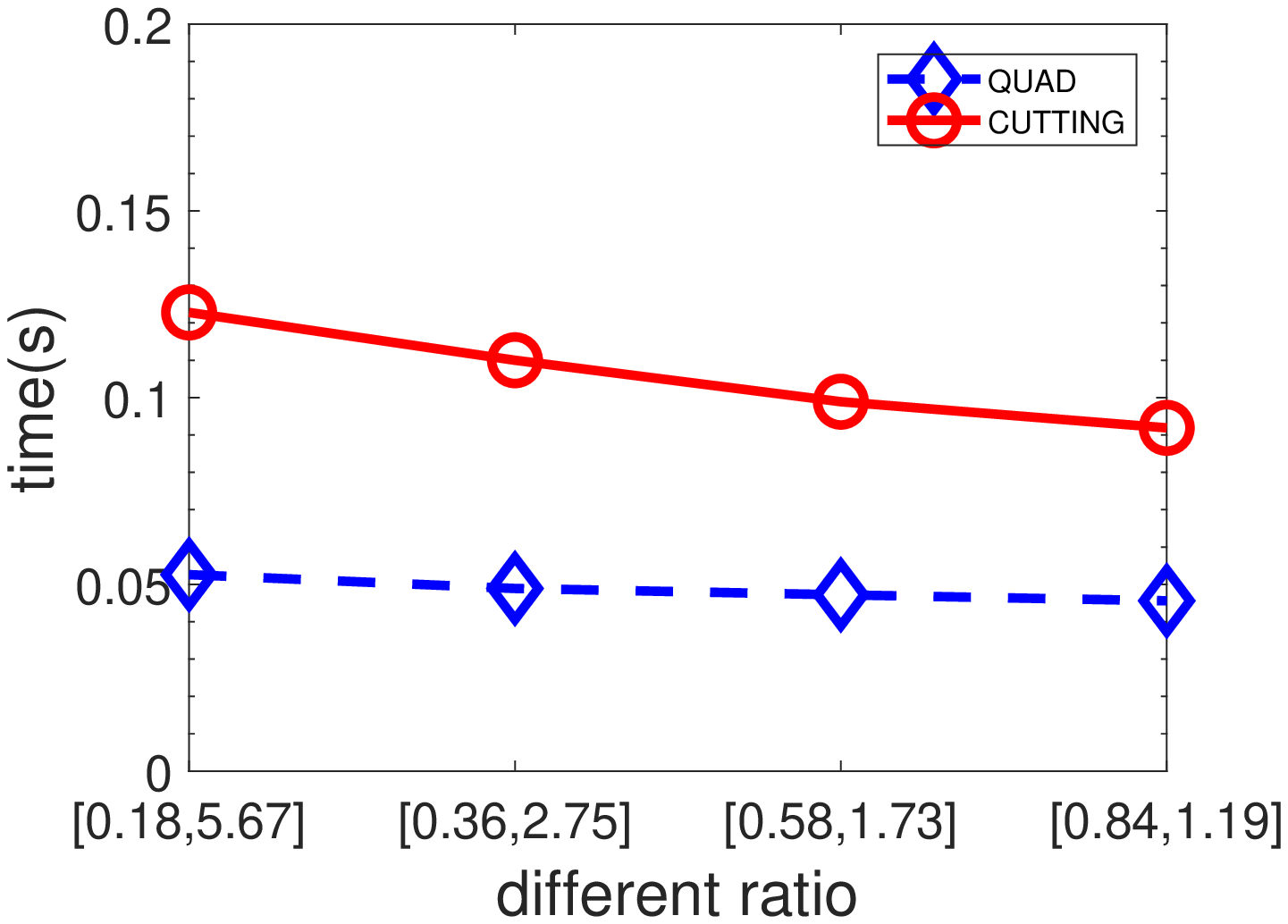}
\end{minipage}
}
\subfigure[\scriptsize{time cost of ANTI}]{
\begin{minipage}[b]{0.22\textwidth}
\includegraphics[width=1.1\linewidth]{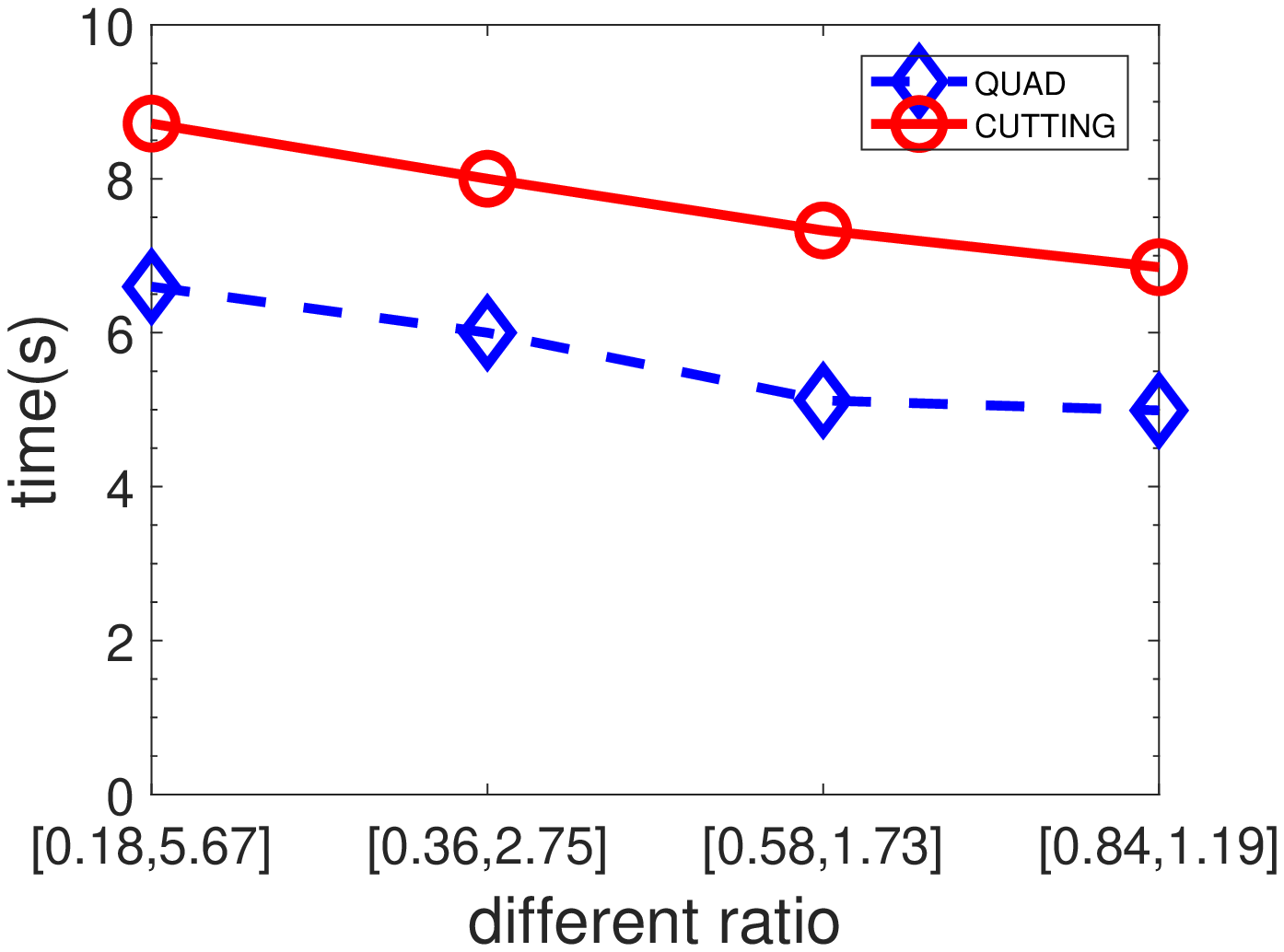}
\end{minipage}
}
\subfigure[\scriptsize{time cost of NBA}]{
\begin{minipage}[b]{0.22\textwidth}
\includegraphics[width=1.1\linewidth]{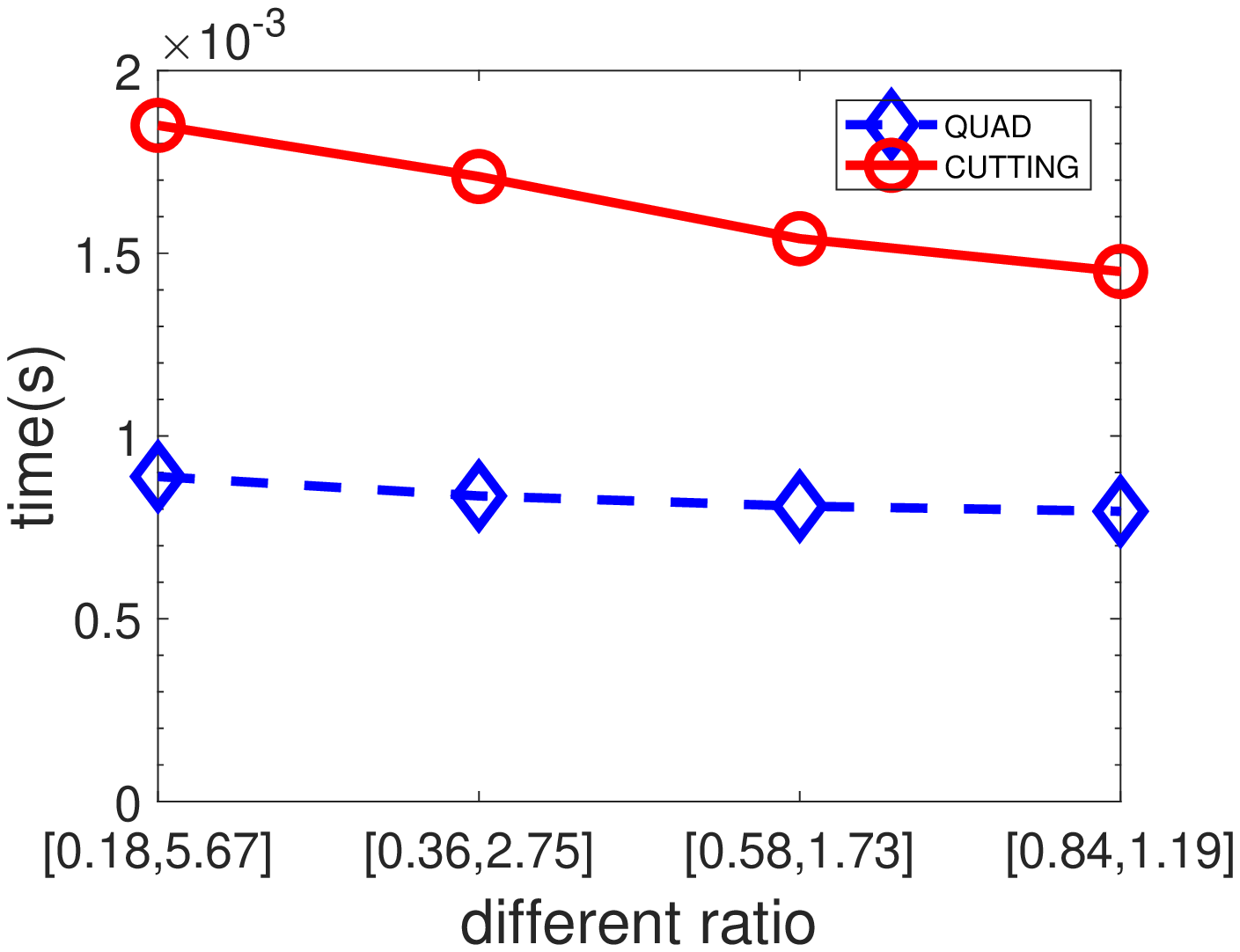}
\end{minipage}
}
 \vspace{-1em}\captionsetup{font={scriptsize}}\caption{The impact of ratio $r$.} \label{fig:rAverage}
\end{figure*}

\begin{figure}[t]
\begin{minipage}[t]{0.22\textwidth}
\centering
\includegraphics[width=1.1\linewidth]{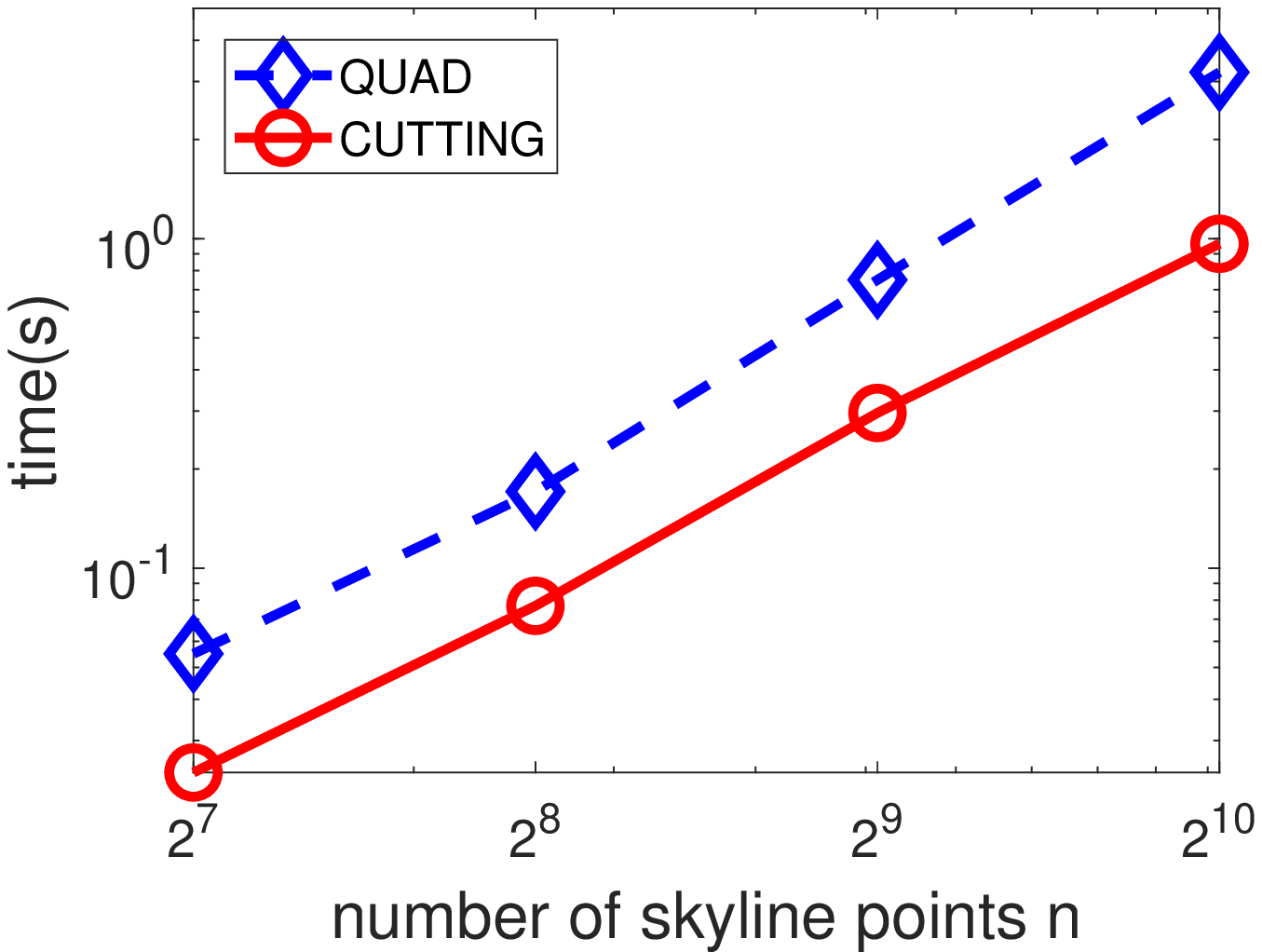}
\captionsetup{font={scriptsize}}
\caption{Worst case ($d=3$).} \label{fig:worstN}
\end{minipage}%
\begin{minipage}[t]{0.22\textwidth}
\centering
\includegraphics[width=1.1\linewidth]{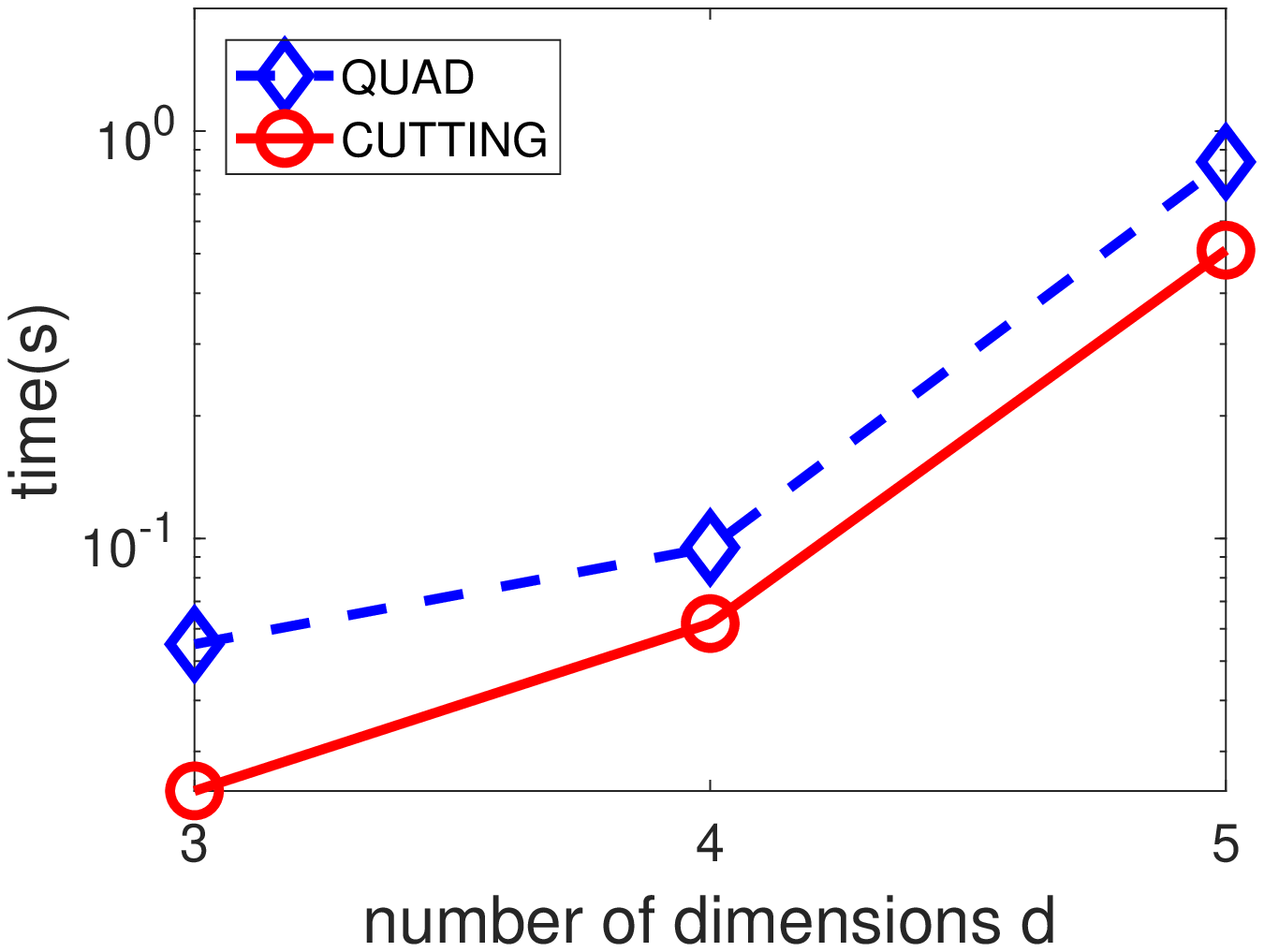}
\captionsetup{font={scriptsize}}
\caption{Worst case ($n=2^7$).} \label{fig:worstD}
\end{minipage}
\end{figure}

\subsection{Performances in the Average Case}
In this subsection, we report the experimental results for our proposed four algorithms in the average case.

\partitle{The impact of $n$}
Figures \ref{fig:nAverage}(a)(b)(c)(d) present the time cost of BASE, TRAN, QUAD, and CUTTING with varying number of points $n$ for the three synthetic datasets and the real NBA dataset ($d=3, r[j]\in [0.36,2.75]$). For the baseline algorithm BASE and the transformation-based TRAN, TRAN is significantly faster than BASE, especially on the ANTI dataset. The reason is that BASE is very sensitive to the number of eclipse points and there are more eclipse points for the ANTI dataset. For the index-based algorithms QUAD and CUTTING, QUAD outperforms CUTTING because the index structure of QUAD is much simpler and it is easier to find the intersecting hyperplanes. Comparing all the different algorithms, the index-based algorithms significantly outperform BASE and TRAN, which validates the benefit of our index structures, especially for large n. From the viewpoint of different datasets, the time cost is in increasing order for CORR, INDE, and ANTI, due to the increasing number of eclipse points. %The time difference between the index-based algorithms and the transformation-based algorithms is significantly large on the NBA dataset, which suggests the NBA dataset is strongly correlated.

\partitle{The impact of $d$}
Figures \ref{fig:dAverage}(a)(b)(c)(d) present the time cost of  BASE, TRAN, QUAD, and CUTTING with varying number of dimensions $d$ for the three synthetic datasets ($n=2^{10}, r[j]\in [0.36,2.75]$) and the real NBA dataset ($n=1000, r[j]\in [0.36,2.75]$). For the non-index-based algorithms BASE and TRAN, TRAN is significantly faster than BASE. For the index-based algorithms QUAD and CUTTING, QUAD significantly outperforms CUTTING, especially for large $d$. The reason is that it is time-consuming to find the intersecting Voronoi hypercells because the number of the vertexes of each Voronoi hypercell in high dimensional space is too high, while it is very easy to find the intersecting subquadrants in QUAD. Comparing all the different algorithms, QUAD and CUTTING significantly outperform BASE and TRAN again, which validates the benefit of our index structures. Because QUAD and CUTTING employ the same binary search tree structure in two dimensional space, they have the same performance. %The time difference between the index-based algorithms and the non-index-based algorithms is significantly large on the NBA dataset, which suggests the NBA data is strongly correlated again.

\partitle{The impact of ratio $r$}
Figures \ref{fig:rAverage}(a)(b)(c)(d) present the time cost of QUAD and CUTTING with varying attribute weight ratio vectors for the three synthetic datasets ($n=2^{10}, d=3$) and the real NBA dataset ($n=1000, d=3$). Because the attribute weight ratio vector has no impact on the transformation-based algorithms, we did not report the time cost for them. For the index-based algorithms QUAD and CUTTING, QUAD significantly outperforms CUTTING again. We can observe that the larger ratio range, the higher time cost for both index-based algorithms. The reason is that we need to search for more space for the larger ratio range, which means that we need to compute more intersections.

\subsection{Performances in the Worst Case}
QUAD has a good performance in the average case, however, QUAD is very sensitive to the dataset distribution because line quadtree can be quite unbalanced if all the lines almost lie in the same quadrant in each layer. We did the experiment with the scenario where all the lines almost lie in the same quadrant. Because the index-based algorithms significantly outperform the transformation-based algorithms, in this subsection, we only report the experimental results for the index-based algorithms in the worst case.

Figure \ref{fig:worstN} presents the worst case time cost of QUAD and CUTTING on different number of points $n$. It is easy to see that CUTTING always outperforms QUAD because the depth for the line quadtree index structure is $O(n)$ in the worst case, i.e., we need to scan all the lines.

Figure \ref{fig:worstD} presents the worst case time cost of QUAD and CUTTING on different number of dimensions $d$. It is easy to see that CUTTING outperforms QUAD again, however, with the increasing of the number of dimensions $d$, the difference becomes small because there are too many vertexes for each Voronoi hypercell in high dimensional space.

%---------------------------------------------------------------------------------------------------------------------------------------

\section{Related Work}\label{sec:Related}

The $k$NN query is the most well-known similarity query. Furthermore, $k$NN has been extensively used in the machine learning field, such as $k$NN classification \cite{DBLP:journals/tit/CoverH67} and $k$NN regression \cite{altman1992introduction}. The most disadvantage of $k$NN query is that it is hard to set the exact and appropriate attribute weight vector.

Skyline is a fundamental problem in computational geometry because skyline is an interesting characterization of the boundary of a set of points. Since the introduction of the skyline operator by Borzsonyi et al. \cite{DBLP:conf/icde/BorzsonyiKS01}, Skyline has been extensively studied in database. \cite{DBLP:journals/ipl/LiuXX14} presented state-of-the-art algorithm for computing skyline from the theoretical aspect. To facilitate the skyline query, skyline diagram was defined in \cite{DBLP:conf/icde/LiuY0PL18}. To protect data privacy and query privacy, secure skyline query was studied in \cite{DBLP:conf/icde/LiuY0P17}. \cite{DBLP:journals/tkde/CuiCXLSX09} studied the skyline in P2P systems. \cite{DBLP:conf/vldb/PeiJLY07, DBLP:conf/cikm/LiuZXLL15, DBLP:journals/www/ZhangLCZC15} studied the skyline on the uncertain dataset. \cite{DBLP:conf/ssdbm/LuZH10} studied the continuous skyline over distributed data streams. \cite{DBLP:journals/pvldb/LiuXPLZ15, DBLP:conf/cikm/YuQL0CZ17} generalized the original skyline definition for individual points to permutation group-based skyline for groups. \cite{ehrgott2005multicriteria} detailedly discussed the multicriteria optimization problem subjected to a weighted sum range. \cite{zitzler2002performance} presented the quantitative comparison of the performance of different approximate algorithms for skyline.

The drawback of skyline is that the number of skyline points can be prohibitively high. There are lots of existing works trying to alleviate this problem \cite{DBLP:conf/icde/LinYZZ07, DBLP:conf/icde/TaoDLP09, DBLP:conf/icde/SarmaLNLX11, DBLP:journals/vldb/MagnaniAM14, DBLP:conf/edbt/SoholmCA16, DBLP:journals/tkde/BaiXWZZYW16, DBLP:journals/tkde/LuJZ11, DBLP:journals/vldb/ZhangLOT10}. Lin et al. \cite{DBLP:conf/icde/LinYZZ07} studied the problem of selecting $k$ skyline points to maximize the number of points dominated by at least one of these $k$ skyline points. Tao et al. \cite{DBLP:conf/icde/TaoDLP09} proposed a new definition of representative skyline that minimizes the distance between a non-representative skyline point and its nearest representative. Sarma et al. \cite{DBLP:conf/icde/SarmaLNLX11} formulated the problem of displaying $k$ representative skyline points such that the probability that a random user would click on one of them is maximized. Magnani et al. \cite{DBLP:journals/vldb/MagnaniAM14} proposed a new representative skyline definition which is not sensitive to rescaling or insertion of non-skyline points. Soholm et al. \cite{DBLP:conf/edbt/SoholmCA16} defined the maximum coverage representative skyline which maximizes the dominated data space of $k$ points. Lu et al. \cite{DBLP:journals/tkde/LuJZ11} proposed the top-$k$ representative skyline based on skyline layers. Zhang et al.\cite{DBLP:journals/vldb/ZhangLOT10} showed a cone dominance definition which can control the size of returned points. All those works are focused on static data, while Bai et al. \cite{DBLP:journals/tkde/BaiXWZZYW16} studied the representative skyline definition over data streams.

Another related direction to our work is dynamic preferences \cite{DBLP:journals/tkde/WongPFW09, DBLP:journals/pvldb/WongFPHWL08, DBLP:conf/kdd/WongFPW07}. Taking the weather as an example, we only consider sunny and raining here. We may prefer sunny for hiking but we may prefer raining for sleeping because raining makes you feel more comfortable. Mindolin et al. \cite{DBLP:journals/pvldb/MindolinC09, DBLP:journals/vldb/MindolinC11} proposed a framework ``p-skyline'' which incorporates relative attribute importance in skyline allows for reduction in the corresponding query result size.

In this paper, we formally generalize the $1$NN and skyline queries with eclipse using the notion of dominance. The most related to our work is \cite{DBLP:journals/pvldb/CiacciaM17}. They defined a query similar to the eclipse query we initially presented in \cite{DBLP:journals/corr/Liu0ZL17} without studying the formal properties with respect to the domination. In addition, they only presented one algorithm for which we formally proved the correctness and used as a baseline algorithm.  We then presented significantly more efficient transformation-based algorithms and new index-based eclipse query algorithms for computing eclipse points.

%---------------------------------------------------------------------------------------------------------------------------------------

\section{Conclusion}\label{sec:Conclusion}
In this paper, we proposed a novel eclipse definition which provides a more flexible and customizable definition for the classic $1$NN and skyline. We first illustrated a baseline $O(n^22^{d-1})$ algorithm to compute eclipse points, and then presented an efficient $O(n\log^{d-1} n)$ algorithm by transforming the eclipse problem to the skyline problem. For different users with different attribute weight vectors, we showed how to process the eclipse query based on the index structures and duality transform in $O(u+m)$ time. A comprehensive experimental study is reported demonstrating the effectiveness and efficiency of our eclipse algorithms.

%---------------------------------------------------------------------------------------------------------------------------------------------------------------------------------------------------------------------------------------------------------------
%---------------------------------------------------------------------------------------------------------------------------------------------------------------------------------------------------------------------------------------------------------------

%---------------------------------------------------------------------------------------------------------------------------------------------------------------------------------------------------------------------------------------------------------------

%-----------------------------------------------------------------------------------------------------------------------------------------
{\tiny
\bibliographystyle{abbrv}\tiny
\bibliography{Eclipse}
}
\end{document}